%% file: main.tex
\documentclass[11pt]{article}
\usepackage{graphicx} 
\usepackage{amsmath,amsthm}
\usepackage{xcolor}
\usepackage[margin=1in]{geometry}
\usepackage{algorithmicx}
\usepackage{algorithm} 
\usepackage{algpseudocode}
\usepackage{xifthen}
\usepackage{graphicx}
\usepackage{wrapfig}
\usepackage{dblfloatfix}
\usepackage{cleveref}
\usepackage{paralist}
\usepackage{xspace}
\usepackage{comment}
\usepackage{url}
\usepackage{authblk}
\excludecomment{ignore}

\newtheorem{corollary}{Corollary}
\newtheorem{theorem}{Theorem}
\newtheorem{lemma}{Lemma}
\newtheorem{definition}{Definition}
\newtheorem{observation}{Observation}

\newcommand{\Faith}[1]{{\color{orange} \textbf{Faith: }#1}}
\newcommand{\Ayaz}[1]{{\color{cyan} \textbf{Ayaz: }#1}}

\newcommand{\prediction}{prediction\xspace} 
\newcommand{\predictions}{predictions\xspace}
\newcommand{\honest}{honest\xspace}

\newcommand{\classification}{classification\xspace}

\newcommand{\threshMisclassifiedVar}{k}     
\newcommand{\actualMisclassifiedVar}{k_{A}} 
\newcommand{\faultyMisclassifiedVar}{k_{F}}

\newcommand{\allErrorsVar}{B}           
\newcommand{\faultyErrorsVar}{B_{F}}    

\algnewcommand{\BlueComment}[1]{\textcolor{blue}{\hfill\(\triangleright\) #1}}
\title{Byzantine Agreement with 
Predictions}
\author[1]{Naama Ben-David}
\author[2]{Muhammad Ayaz Dzulfikar}
\author[3]{Faith Ellen}
\author[2]{Seth Gilbert}
\affil[1]{Technion}
\affil[2]{National University of Singapore}
\affil[3]{University of Toronto}

\begin{document}

\begin{titlepage}
    
\maketitle
\input{abstract}

\end{titlepage}

\input{introduction}

\input{model}

\input{techniques}

\input{highlevel}

\input{classification}
\input{faith-unauthenticated}
\input{authenticated_v2}
\input{final_algo_proofs}
\input{lower-bounds}

\input{conclusion}

\bibliographystyle{acm}
\bibliography{references}


\end{document}

%% file: abstract.tex
\begin{abstract}    
    In this paper, we study the problem of \emph{Byzantine Agreement with predictions}.  Along with a proposal, each process is also given a prediction, i.e., extra information which is not guaranteed to be true.  For example, one might imagine that the prediction is produced by a network security monitoring service that looks for patterns of malicious behavior.
    
    Our goal is to design an algorithm that is more efficient when the predictions are accurate, degrades in performance as predictions decrease in accuracy, and still in the worst case performs as well as any algorithm without predictions even when the predictions are completely inaccurate.

    On the negative side, we show that Byzantine Agreement with predictions still requires $\Omega(n^2)$ messages, even in executions where the predictions are completely accurate. On the positive side, we show that \emph{classification predictions} can help improve the time complexity.  For (synchronous) Byzantine Agreement with classification predictions, we present new algorithms that leverage predictions to yield better time complexity, and we show that the time complexity achieved is optimal as a function of the prediction quality.
\end{abstract}


%% file: introduction.tex
\section{Introduction}
\label{section:intro}

\newcommand{\para}[1]{\vspace{0.5em}\noindent\textbf{#1}}



Imagine you are in charge of the security for a critical distributed system storing sensitive data, such as a financial application, a cryptocurrency exchange, or a healthcare management system.  It is quite likely that you would install a product like Darktrace~\cite{darktrace}, Vectra AI~\cite{vectraai}, or Zeek~\cite{zeek} (among many others) to monitor the network and flag suspicious behavior.  Many of these tools use AI or machine learning techniques to identify user patterns that typify malicious activity.  The big picture we are asking in this paper is: \emph{Can distributed algorithms productively use the information produced from these (AI-based) security monitoring systems to design more efficient algorithms?}

\para{Algorithms with predictions.} There has been 
a significant interest of late in the area of ``algorithm with predictions,'' also known as ``algorithms with advice'' or ``learning-augmented algorithms'' (see, e.g.,~\cite{mitzenmacher2020algorithmspredictions,mitzenmacher22predictionscomm}).  The basic idea is that some external module, likely based on AI or machine learning techniques, can provide algorithms with some predictive information.  These predictions might be based on historical data or other data unavailable to the algorithm designer.  (And the prediction software can evolve and improve independently from the algorithm in question.)  A classic example in this paradigm is searching for data in an array (or dictionary): we know that in the worst case, it will take $\Theta(\log{n})$ time to perform a binary search; but if we are given a ``prediction'' as to approximately where to start our search, we can complete the search significantly faster.  Other examples include learning-augmented algorithms for the ski rental problem~\cite{ANGELOPOULOS24onlineuntrustedadvice,gollapudi19aski,purohit18onlineviaml}, load balancing~\cite{ahmadian23loadbalancing}, and distributed graph algorithms~\cite{boyar2025distributedgraphalgorithmspredictions}. 

\para{Byzantine agreement.} In this paper, we consider the problem of Byzantine agreement where a collection of processes, some of which may be faulty or malicious, are trying to agree on a value.  The problem of agreement lies at the heart of many secure distributed systems such as blockchains, distributed databases, distributed file systems, distributed lock servers, etc.   

The key problem faced by these types of secure distributed systems is that Byzantine agreement is inherently expensive: in a (synchronous) system with $n$ processes, up to $t$ of which may be faulty, any (deterministic) algorithm requires at least $\Omega(t)$ time~\cite{DRS90} and $\Omega(n^2)$ messages \cite{dolev1985bounds}.

The question we address in this paper is whether we can reduce the cost of Byzantine agreement when devices are equipped with ``predictions'' that provide them useful information about the system, e.g., the type of information that security networking monitoring might produce.  More specifically, we model predictions as per-process input strings that provide some information about the execution, but may not be completely accurate. 

\para{Prediction accuracy.} If the predictions are completely accurate, then the problem is easy.  For example, if each honest process is given a ``prediction'' that completely identifies which other processes are malicious (and all the honest processes share the same accurate information), then the honest processes can simply execute Byzantine agreement among themselves---because there are no malicious processes present, it can complete quickly and efficiently.\footnote{More trivially, if the prediction is perfectly accurate, then it could even specify precisely what value to decide---at least in a model where the prediction can depend on the input. None of our results, either upper or lower bounds, depend on the distinction of whether the advice can be a function of the input.}  Conversely, when the predictions are totally wrong (e.g., casting suspicion on all the honest users and declaring all the malicious users to be honest), it seems nearly impossible to use the predictions to any benefit.

Unfortunately, we cannot assume that the predictions are always completely accurate.  For example, network monitoring software can certainly produce erroneous suspicions, and can sometimes fail to detect malicious behavior; these types of security systems tend to work very well in detecting typical attacks, but can fail when they are confronted by a novel type of attack or atypical behavior.\footnote{One of the authors of this paper has had the personal experience of being suspected of illicit network behavior by eBay, despite doing nothing more than browsing auction listings.}  

As such, our goal is to provide improved performance when the predictions are accurate, and to perform no worse than existing algorithms when the predictions are inaccurate.  Moreover, we want the performance to degrade gracefully with the quality of the prediction: a ``mostly good'' prediction should yield ``mostly good'' performance. 

\para{Our results.}  We first consider message complexity.  One might expect that in executions where the predictions are completely accurate, it should be easy to agree efficiently.  Perhaps surprisingly, we show that (potentially unreliable) predictions yield no improvement in message complexity---even in executions where the predictions are completely accurate!  That is, we show that there is still a lower bound of $\Omega(n^2)$ messages, extending the technique developed by Dolev and Reischuk~\cite{dolev1985bounds} (and generalized in several papers since then, e.g.,~\cite{abraham2019communication,abraham2023eclips,civit2023validity}).\footnote{In some ways, this result implies that verifying the output of an agreement protocol has the same message complexity as solving agreement!  In retrospect, this was also hinted at in lower bounds for graded consensus implied by \cite{abraham2023eclips}.}  More specifically, we show:

\begin{theorem}[Restated from \Cref{theorem:message_lower_bound}]
    For every deterministic algorithm that solves Byzantine Agreement with predictions in a synchronous system with at most $t$ faulty processes, there is some execution in which the predictions are $100\%$ correct and at least $\Omega(n + t^2)$ messages are sent by correct processes.
\end{theorem}

We then turn out attention to time complexity.  To this end, we consider a specific type of \prediction, which we call \emph{\classification \prediction}: each process is provided one bit of information for each other process in the system indicating whether that process is predicted to be faulty or \honest in this execution. Notice that in a system with $n$ processes, this consists of $n^2$ bits of information.  (In the motivating context of security monitoring software, this is equivalent to the software identifying a list of processes that appear malicious, with the default assumption that the remainder are honest.)

We assume that out of the $n^2$ bits of predictive information, at most $B$ of those bits are ``incorrect.'' In the context of security monitoring software, this is equivalent to saying that there are a total of at most $B$ false positives or false negatives, i.e., honest processes that are incorrectly suspected or malicious processes that are missed by the detection mechanism.  

We then give an algorithm that solves Byzantine agreement with predictions.  The algorithm is deterministic, and we assume a synchronous system with $n$ processes, up to $t < n/3$ failures, and an execution where at most $f \leq t$ processes are Byzantine.  Our algorithm runs in $O(\min(B/n+1, f))$ time, when $B = O(n^{3/2})$; otherwise, when $B$ is larger than $\Theta(n^{3/2})$, the algorithm completes in time $O(f)$.  This algorithms has message complexity $O(n^2\log(n))$.  This is summarized as follows:

\begin{theorem}[Restated from \Cref{theorem:unauth_ba_predicitions}]
    There is an unauthenticated algorithm that solves Byzantine Agreement with predictions when there are at most $t < n/3$ faulty processes. Given a classification prediction with $B$ incorrect bits:
    \begin{itemize}
        \item If $B = O(n^{3/2})$, every honest processes decide in $O(\min\{(B/n)+1, f\})$ rounds and in total, honest processes send $O(n^2 \log(\min\{B/n, f\}))$ messages.
        \item Otherwise, every honest processes decide in $O(f)$ rounds and in total, honest processes send $O(n^2 \log(f))$ messages. 
    \end{itemize}
\end{theorem}

The natural question remains whether we can still gain a benefit from predictions even when $B = \Omega(n^{3/2})$.  We show that in an authenticated system---specifically, one with signatures---we can take advantage of whatever accuracy in the predictions is available.  In this case, we assume that $t < (1/2-\epsilon)n$.  We give an algorithm runs in time $O(\min\{(B/n)+1, f)\}$ for all values of $B$, and has message complexity $O(n^3\log{n})$.  This is summarized as follows:

\begin{theorem}[Restated from \Cref{theorem:auth_ba_predicitions}]
There is an authenticated algorithm that solves Byzantine Agreement with predictions when there are at most $t < (1/2-\epsilon)n$ faulty processes. Given a classification prediction with $B$ incorrect bits, every honest processes decide in $O(\min\{B/n, f\})$ rounds and honest processes send $O(n^3 \log (\min\{B/n, f\}))$ messages.
\end{theorem}

Finally, we show that for classification predictions, these time bounds are tight, i.e., any Byzantine agreement protocol with classification predictions requires $\Omega(\min\{B/n+1, f\})$ rounds.  Specifically, we show:
\begin{theorem}[Restated from \Cref{theorem:rndLowerBound}]
    For every deterministic Byzantine agreement algorithm with classification predictions
    and for every $f \leq t < n-1$, there is an execution with $f$ faults in which at least
    $\min\{f+2,t+1, \lfloor B/(n-f) \rfloor+2,\lfloor B/(n-t) \rfloor+1\}$ rounds are performed.
\end{theorem}

\section{Related Work}

The problem of Byzantine agreement has been long studied, first introduced by Pease, Shostak and Lamport~\cite{LamportSP82,PeaseSL80}.  Soon thereafter, it was shown to be inherently expensive, with a lower bound of $t+1$ rounds~\cite{dolev1983authenticated,FischerL82} and $\Omega(nt)$ messages~\cite{dolev1985bounds}.  

One key focus for developing more efficient Byzantine Agreement protocols has been the idea of \emph{early stopping}: in executions with only $f < t$ failures, we can terminate in $f+2$ rounds, which can be much faster in practice.  Chandra and Toueg~\cite{ChandraT90} introduced the idea of early stopping in the context of crash failures, and Dolev, Reischuk, and Strong~\cite{DRS90} gave the first round-optimal early stopping Byzantine agreement protocol.  

The early stopping protocol by Lenzen and Sheikholeslami~\cite{lenzen2022} was a significant inspiration to this paper, giving an improved early stopping protocol that is both time and communication efficient.  A notable aspect of their paper is the idea of using a \emph{validator} (a version of graded consensus) to check agreement before and after a ``king'' tries to lead agreement.  

Graded consensus is a key existing subprotocol that we rely on to ensure validity and to check for agreement.  It is a weaker form of Byzantine agreement which guarantees agreement only if all the honest processes already agree.  It returns a value and a \emph{grade}, and ensures that if all the honest processes start with the same value, then they all return that value with a grade of 1; if any honest process returns a value with a grade of 1, then every honest process returns that value (with a grade of 0 or 1).  Graded consensus was first introduced in the form of adopt-commit protocols~\cite{gafni98}, and has been used frequently ever since (e.g., recently in the Byzantine context to yield more efficient protocols~\cite{abraham23roundcrusader, abraham22bindingcrusader}).  Of note for this paper are two message-efficient variants: unauthenticated~\cite{civit2024efficient} and authenticated~\cite{momose2021}.

Over the last several years, there has been significant interest in the idea of algorithms with (unreliable) predictions (see, e.g.,~\cite{mitzenmacher22predictionscomm,predictionsworkshop}).  In the context of sequential algorithms, a large variety of problems have been considered in this framework, including Bloom Filters~\cite{mitzenmacher2020algorithmspredictions}, ski rental~\cite{purohit18onlineviaml}, load balancing~\cite{0001X21}, and on-line graph algorithms~\cite{AzarPT22}.  See, e.g., ~\cite{BoyarFKLM17,mitzenmacher2020algorithmspredictions} for surveys of more relevant work.

To date, there has been less work on distributed algorithms with (unreliable) predictions.  Recently, there has been progress on distributed graph algorithms with predictions~\cite{boyar2025distributedgraphalgorithmspredictions}.  There has also been work on contention resolution with predictions~\cite{gilbert21contentionprediction}.  These papers do not address Byzantine failures. 

If one focuses on reliable predictions (otherwise known as ``advice''), there is a much larger body of distributed algorithms.  There has been significant work in the distributed computing community on locally checkable labellings (LCLs) (introduced by~\cite{naor95lcl}), and there is a sense in which the ``labels'' can be seen as predictions.  This was further explored in~\cite{Balliu0K0ORS24}, which can be seen for further interesting related work on distributed algorithms with advice.  

One area of seemingly similar research has been work on failure detectors, first introduced by Chandra and Toueg~\cite{chandra96unreliabledetector}. Malkhi and Reiter~\cite{malkhi97intrusiondetection} generalized this notion to Byzantine failure detectors, which was further explored in many follow-up papers (e.g.,~\cite{BALDONI2003185,doudou90mutenessdetector,kihlstrom03byzantinedetectorconsensus}).  Failure detectors have a similar motivation in that they observe it may be possible to detect Byzantine behavior (and explore methods for accomplishing that).  Unlike in the ``predictions'' model, they typically assume that the failure detectors are eventually perfect, over time updating their predictions until they correctly identify every faulty process.  (E.g.,~\cite{malkhi97intrusiondetection} includes a property of ``strong completeness'' where ``eventually every quiet process is permanently suspected by every correct process.'')
By contrast, we are primarily interested in the case where predictions may never be perfectly complete or accurate, and where performance degrades gracefully with prediction quality. 
    
Another related area of research has been the idea of \emph{accountability}, first introduced in distributed systems by~\cite{haeberlen2007peerreview}.  There has been significant work of late on accountable Byzantine agreement~\cite{civit21polygraph,civit22abc,CivitGGGKMS22,ShengWNKV21}.  The goal of these papers is to develop techniques for detecting malicious processes---not how to use such predictions to generate efficient protocols.

%% file: model.tex
\section{Model}

\para{Processes.} We consider a system with $n$ processes $\{p_1, \dots, p_n\}$.  Up to $t$ of these processes may be faulty, and faulty processes may act in an arbitrary (``Byzantine'') manner.
The remaining processes are honest.  For a given execution, let $f \leq t$ be the actual number of faults.  (The parameter $f$ is not known to the protocol, while $t$ and $n$ are.) We use $H$ to denote the set of honest processes' identifiers, and $F = \{1, \dots, n\} - H$ as the set of faulty processes' identifiers.

\para{Synchronous Network.} The processes are connected by a synchronous network, and the algorithm is executed in a round-by-round manner: in each round, each process may transmit messages to other processes, receive messages, and update its state.

\para{Predictions.} Each process $p_i$ also receives as a prediction a binary string $a_i$.  Each process may receive a different string. 

We consider a specific type of predictive information known as \emph{classification predictions}.  In that case, each $a_i$ is an $n$-bit string where
\begin{itemize}
    \item[] $a_i[j] = 0$ means that $p_j$ is predicted to be faulty, and
    \item[] $a_i[j] = 1$ means that $p_j$ is predicted to be honest.
\end{itemize}

A prediction bit $a_i[j]$ is correct if the reality matches the implied prediction; otherwise, it is incorrect.  More specifically, for a given execution: 
The number of bits that incorrectly predict a faulty process as honest is $B_F = \sum_{i\in H} \#\{j \in F\ |\ a_i[j] =1\}$; the number of bits that incorrectly predict an honest process as
faulty is $B_H = \sum_{i\in H} \#\{j \in H\ |\ a_i[j] = 0\}$.  The total number of incorrect bits in the \prediction{} is $B = B_F + B_H$.  (Note that incorrect bits in the \prediction{} given to faulty processes are not counted.) 

\para{Byzantine Agreement.} 
The problem of Byzantine Agreement is defined as follows: each process $p_i$ begins with an input value $v_i$ and eventually produces a decision satisfying:
\begin{itemize}
    \item \emph{Agreement}: All honest processes decide the same value.
    \item \emph{Strong Unanimity (Validity)}: If every honest process has the same input value $v$, then they all decide $v$.
    \item \emph{Termination}: Eventually every honest process decides.
\end{itemize}

The time complexity (or round complexity) of a Byzantine Agreement protocol is the number of rounds it takes until the last honest process decides.  The message complexity is the total number of messages sent by honest processes up until they decide.  (In both cases, we are interested in the worst-case over all possible inputs and all possible patterns of Byzantine behavior.)

%% file: techniques.tex
\section{Technical Overview of the Algorithms}

In this section, we provide an overview of the technical results, focusing on the algorithms.  Lower bounds can be found in \Cref{sec:lowerbounds}.

\subsection{Classification}

\newcommand{\floor}[1]{\lfloor #1 \rfloor}
\newcommand{\ceil}[1]{\lceil #1 \rceil}

The agreement protocols presented rely on predictions, and they only succeed if honest processes trust the same processes.  Thus, we need their predictions to be (somewhat) similar to each other in order to guarantee success. However, even with relatively few error bits in the input \predictions $a_i$, the sets of locally trusted processes may diverge from each other by quite a bit. 
Thus, the first key ingredient in our algorithms is using the prediction strings $a_i$ to classify the processes as honest or faulty.   This protocol is presented in \Cref{section:classification}.

To that end, we rely on voting: each process shares its predictions with the others; process $p_i$ classifies $p_j$ as honest if it receives at least $\ceil{(n+1)/2}$ votes that $p_j$ is honest; otherwise, it classifies $p_j$ as faulty.  

We prove that after this voting phase, the newly created classifications at each honest process will not diverge too much from each other, and will not include too many misclassified processes.   Specifically, recall that $B$ is the number of total erroneous prediction bits over all \prediction strings of all honest processes in a given execution.  This classification ensures that there are at most $O(B/n)$ processes that are misclassified by at least one honest process (see \Cref{Flemma:advice_preprocess_misclassify_bound}).  That is, if $M_i$ is the set of processes misclassified by process $p_i$, then $|\cup M_i| = O(B/n)$.  (Note each misclassified process is counted only once, even if it is misclassified by multiple processes.)

We typically use the variable $k_A = |\cup M_i|$ to denote the number of \emph{classification errors}.  This approximate agreement on a classification and the bound on classification errors allows sufficient coordination to solve Byzantine agreement efficiently.


\subsection{Adapting to the Quality of Classification}

A key requirement is an algorithm that is always correct, that works well when predictions are good, and that degrades in performance as the predictions get worse. 

To achieve this, we are going to design conditional Byzantine agreement algorithms that guarantee good performance in executions where the number of classification errors are at most some fixed $k$ (and provide minimal guarantees otherwise).   We present two conditional Byzantine agreement algorithms, one that does not use any authentication, and another that improves upon the performance of the first, but makes use of cryptographic signatures to do so.  


We will then use the typical guess-and-double approach, executing phases where $k = 1, 2, 4, \ldots, $.  In each phase, we perform one instance of the conditional Byzantine agreement protocol; we also run $O(k)$ rounds of a standard \emph{early-stopping} Byzantine agreement protocol.  The early-stopping protocol ensures that if the number of actual failures is $< k$, then we will also terminate.  This prevents bad predictions from delaying the protocol longer than necessary.

To ensure validity (i.e., strong unanimity) and consistency among these various partial and conditional Byzantine agreement protocols, we use  \emph{graded consensus} protocol, which we discuss in more detail in~\Cref{sec:highlevel}. 

The high-level ``guess-and-double'' wrapper protocol can be found in \Cref{sec:highlevel}.

\subsection{Conditional Byzantine Agreement with Classification: Unauthenticated}

The first conditional Byzantine Agreement protocol we give does not use any authentication.  The details can be found in \Cref{section:un-byz-adv}, and the protocol is listed in \Cref{algo:un_ba_class}. The algorithm goes through phases, each with some leaders.

The protocol receives as input a proposed value $v_i$, a classification $c_i$ (resulting from the voting protocol), and an upper bound $k$ on the maximum number of classification errors to be tolerated.

Each process uses its (local) classification to determine how to prioritize potential leaders---ordering all the processes predicted to be honest before those that are predicted to be faulty, and then sort them by their identifiers.

The protocol proceeds through a set of phases, each of which relies on a different subset of the potential leaders, chosen in order of priority.  Since processes do not agree exactly on the classification, their prioritization of leaders will be different.  Thus we cannot trust a single leader at a time, but rather we will choose a set of $\Theta(k)$ leaders in each round.  Of course, each (honest) process may have a different leader set.  

In each round, a ``conciliation'' subroutine is executed that attempts to help the processes converge on a single proposal.  The conciliation routine guarantees agreement when sufficiently many processes trust the same honest leaders in that round. More specifically, conciliation will succeed when: (i) all the honest processes choose honest leader sets, and (ii) there is a set of $2k+1$ honest processes that are included in the leader set of every honest process.  (These conditions may seem unlikely to occur, but in fact follow from the classification and the prioritization of leaders.)  During the conciliation phase, processes compute a ``leader graph'' to estimate the overlap in leader sets, compute a proposal based on the connectivity of the leader graph, and then choose a final proposal based on leader voting. 

The final component of the protocol is to determine when convergence has occurred, i.e., when all the honest processes have the same proposed decision.  At the same time, we have to ensure that strong unanimity holds.

To accomplish this, we follow a similar paradigm as~\cite{lenzen2022}, using a \emph{graded consensus} protocol to detect agreement.  Graded consensus (sometimes referred to as \emph{adopt-commit}~\cite{gafni98}) is a weaker form of agreement that only guarantees complete agreement when all the honest processes initially agree.

More specifically, graded consensus outputs a value and a \emph{grade} $\in\{0, 1\}$, wherein: (i) if all honest processes begin with the same value, then they all output the same value with grade $1$, and (ii) if any honest process outputs a value with grade $1$, then every honest process outputs the same value (with either grade $0$ or grade $1$).  

We therefore run an instance of graded consensus before each conciliation (to ensure strong unanimity, when all processes begin with the same value), and after each conciliation (to detect agreement).

To achieve the desired performance, a small modification of the graded consensus protocol is needed: instead of using all-to-all communication, we rely on the leader set for the phase to drive the graded consensus protocol, reducing the message complexity.

\subsection{Conditional Byzantine Agreement with Classification: Authenticated}

The second conditional Byzantine Agreement protocol we give relies on authentication.  The details can be found in \Cref{section:auth_byz_using_class}, and the protocol is listed in \Cref{algo:auth_ba_class}.

The conditional algorithm, as described in the previous section,  ensures success only when $B = O(n^{3/2})$.  This is because when $B$ is larger, there remains too much disagreement in the classification, and hence we cannot achieve the conditions needed for the conciliation protocol to succeed, i.e., we cannot ensure we quickly reach a point where the honest processes have selected honest leaders.

Instead, by using authentication, we can relax the requirement: honest processes only need to choose leaders where a majority of the leaders are correct.  More specifically, if $k$ is the upper bound on classification errors, each honest process ``votes'' for the $2k+1$ processes it prioritizes the most to be leaders by sending signed messages to them.  Each leader that receives at least $t+1$ votes declares itself to be part of the leader committee---and uses those $t+1$ votes as a certificate that proves it is on the committee. We show that at least $k+1$ honest processes and at most $k$ faulty processes are members of the committee.

At that point, each process on the leader committee performs a ``Byzantine broadcast'' to send its proposed value to the other members of the committee.  This Byzantine broadcast ensures that all members of the committee receive the same set of values, and that it consists of precisely the values broadcast by the committee members.

Finally, the committee members can deterministically choose a decision value from that set and send the value to everyone.  Since processes only accept messages sent by committee members with certificates proving they are committee members, and since a majority of the processes with such certificates are honest, it is safe for processes to adopt the majority value and decide.

One final note is on the implementation of the Byzantine broadcast, which relies on a Dolev-Strong~\cite{dolev1983authenticated} like algorithm of constructing signature chains. The classical algorithm is slightly modified to only accept messages from committee members, and to reduce the running time to $O(k)$ rounds---since we know that there should be at most $k$ faulty members of the committee.

It is interesting to observe the sense in which signatures allow for a more ``robust'' leader committee.  It is important that a process can efficiently ``prove'' that it has been selected as a leader by at least one honest process.  
Furthermore, being able to use signature chains to efficiently ensure that information has been properly disseminated proves to be useful. 
It is not immediately obvious how to accomplish these tasks efficiently (e.g., without increasing the time and message complexity) without signatures.

%% file: highlevel.tex
\section{High-Level Guess-and-Double Strategy}\label{section:final_algo}
\label{sec:highlevel}


In this section, we present the main ``wrapper'' algorithm that transforms the ``conditional'' algorithms presented later in \Cref{section:un-byz-adv} and \Cref{section:auth_byz_using_class} into an algorithm that succeeds in all executions.

\begin{algorithm} [h]
\caption{Byzantine Agreement with Predictions: (code for process $p_i$)}
\label{algo:final_ba}
\footnotesize
\textbf{ba-with-predictions}($x_i,a_i$):
\begin{algorithmic} [1] 
\State $c_i \leftarrow \textbf{classify}(a_i)$
\State $v_i \leftarrow x_i$
\State $\mathit{decided}_i \gets \mathit{false}$
\State \textbf{for} $\phi \gets 1$ to $\lceil\log_2 t\rceil + 1$:
    \State \hskip2em $T \gets \alpha \cdot 2^{\phi-1}$       
    \BlueComment{Choose $\alpha$ so \textbf{ba-early-stopping} and \textbf{ba-with-classification} can complete for this phase.}
    
    \State \hskip2em $(v_i,g_i) \gets$ \textbf{graded-consensus}($v_i$)\label{line:final-ba-gc1}  
    \BlueComment{Graded Consensus protects validity.}
    \State \hskip2em $v'_i \gets$ \textbf{ba-early-stopping}($v_i, T$) \BlueComment{early-stopping BA with time limit of $T$ rounds}
    \State \hskip2em \textbf{if} $g_i = 0$ \textbf{then} $v_i \gets v'_i$\label{line:final-ba-use-early-stoppping}
    
    \State \hskip2em $(v_i,g_i) \gets$ \textbf{graded-consensus}($v_i$)\label{line:final-ba-gc2} 
    \BlueComment{Graded Consensus protects validity.}

    \State \hskip2em $v'_i \gets$ \textbf{ba-with-classification}($v_i, c_i, 2^{\phi-1}, T$) 
    \BlueComment{at most $2^{\phi-1}$ classification errors, time limit $T$}
    \State \hskip2em \textbf{if} $g_i = 0$ \textbf{then} $v_i \gets v'_i$\label{line:final-ba-use-classification}

    \State \hskip2em $(v_i,g_i) \gets$ \textbf{graded-consensus}($v_i$)\label{line:full-ba-gc-last}
    \BlueComment{Graded Consensus checks for agreement.}
    \State \hskip2em \textbf{if} $\mathit{decided}_i$ \textbf{then return}($decision_i$)
    \State \hskip2em \textbf{if} $g_i = 1$ \textbf{then}
        \State \hskip4em $decision_i \gets v_i$
        \State \hskip4em $\mathit{decided}_i \gets \mathit{true}$
    \State\textbf{return}($decision_i$)
\end{algorithmic} 
\end{algorithm}


We will employ two Byzantine agreement protocols: (i) an early-stopping Byzantine agreement protocol with $O(f)$ round complexity, and (ii) a conditional Byzantine agreement with classification that, when combined with the classification mechanism (see \Cref{sec:classification}), has $O(B/n)$ round complexity. 

To ensure validity and consistency, the algorithm relies on a graded consensus sub-protocol, which has been previously developed in~\cite{civit2024efficient} (without authentication) and~\cite{momose2021} (with authentication).  Graded consensus takes an input value $v_i$ and produces an output value and a grade. It provides the following guarantees: 
\begin{compactitem}

    \item \emph{Strong Unanimity}: If every honest process has the same input value, $v$, then they all return $(v, 1)$.
    
    \item \emph{Coherence}: If an honest process returns $(v, 1)$ then every honest process returns the value $v$.
    
    \item \emph{Termination}: All honest processes return simultaneously (i.e., in the same round).
\end{compactitem}
The graded consensus protocols that we will be using run in $O(1)$ rounds and have $O(n^2)$ message complexity.


The algorithm will proceed in 
$\lceil\log_2(t)\rceil + 1$ phases. 
We will ensure that phase $\phi$ completes in $O(2^{\phi-1})$ rounds and
the algorithm terminates one phase after if either of the following conditions is met: (i) the execution has at most $2^{\phi-1}$ faulty processes, or (ii) the classification advice misclassifies at most $2^{\phi-1}$ processes (given $B$ is not too large). Since the classification advice misclassifies at most $O(B/n)$ processes, this immediately yields the desired round complexity of $O(\min\{(B/n)+1, f\})$.

Each phase consists of five parts: two agreement sub-protocols and three graded consensus sub-protocols.  We use a time limit of $T$ to determine how long each agreement sub-protocol can run for---every process synchronously spends $T$ time (no less, no more) on the sub-protocol, aborting it if necessary.\footnote{The protocols presented later do not include the parameter $T$---we assume here they are modified to reflect the time limit.}   Each phase proceeds as follows: 

First, we execute a graded consensus using the existing proposed value.  If all honest processes begin with the same value, everyone will return a grade of $1$ and the  value of an honest process will never change again.
Second, we execute an early stopping Byzantine agreement protocol for $O(2^{\phi-1})$ rounds.  If there are fewer than $O(2^{\phi-1})$ faulty processes, then, by the end of the protocol, every honest process has the same value.
Third, we execute another instance of graded consensus.  Again, if every honest process has the same value, everyone will return a grade of $1$ and no honest process will change its  value again.
Fourth, we execute a conditional Byzantine Agreement with classification.  If there are fewer than $O(2^{\phi-1})$ classification errors, then by the end of the protocol, then every honest process returns the same value.
Fifth, we execute one final instance of graded consensus, again checking if all honest processes have the same value.  If it returns a grade of $1$, then we commit to our decision---and execute one more phase.  

The reason that we run one additional phase is that we cannot be sure that all honest processes have reached a decision in this phase.  However, if an honest process decides in phase $\phi$, we can show that all honest processes will decide by the end of phase $\phi+1$, so it is only necessary to continue to help for one more phase.

In the next three sections, we will present the details of the conditional Byzantine Agremeent protocols.  In \Cref{section:final_algo_proof}, we will combine this algorithm with the conditional algorithms to show the claimed results.

%% file: classification.tex
\section{Classification Predictions}\label{sec:classification}
\label{section:classification}

It is useful to have the honest processes agree, if possible, on which processes they treat as faulty. If sufficiently many  honest processes
have the same prediction about some process, then voting will ensure that they all treat that process the same way.
In  Algorithm \ref{algo:vote}, each honest process $p_i$ broadcasts its prediction $a_i$. Faulty processes may send different
vectors to different honest processes or may fail to send an $n$-bit vector to some honest processes.
For each honest process $p_i$, its \emph{classification} of the processes is the binary vector $c_i$ whose $j$'th component, $c_i[j]$, is 1 ($p_i$ classify $p_j$ as honest) if
$p_i$ receives at least $\lceil (n+1)/2 \rceil$ messages (including from itself) that predicts $p_j$ is honest. Otherwise, $c_i[j] = 0$. 
An honest process $p_i$ \emph{misclassifies} an
honest process $p_j$ if $c_i[j] = 0$. It \emph{misclassifies} a faulty process $p_j$ if $c_i[j] = 1$.
A process $p_j$ is {\em misclassified} if it is misclassified by some honest process.


\begin{algorithm} [h]
\caption{Classification (code for process $p_i$)}
\label{algo:vote}
\footnotesize
\begin{algorithmic} [1] 
\Statex \textbf{classify($a_i$)}
\hskip2em 
\State \textbf{broadcast} $\mathit{a}_i$
\State Let $R_i$ denote the multiset of $n$-bit vectors $p_i$ received (including from itself) 
\State \textbf{for} $j \gets 1$ to $n$ \textbf{do}:
        \State \hskip2em \textbf{if}
        $\# \{a[j] = 1\ |\ a \in R_i\} \geq \lceil (n+1)/2 \rceil$
        \State \hskip2em \textbf{then} $c_i[j] \leftarrow 1$
        \State \hskip2em \textbf{else} $c_i[j] \leftarrow 0$
 \State \textbf{return} $c_i$
\end{algorithmic} 
\end{algorithm}

Fix an execution of Algorithm~\ref{algo:vote}.

\begin{observation}
If a faulty process $p_j$ is misclassified, then at least $\lceil (n+1)/2 \rceil-f$ honest processes received an incorrect prediction about $p_j$.
\label{misclass-faulty}
\end{observation}


\begin{observation}
If an honest process $p_j$ is misclassified, then 
at least $\lceil n/2 \rceil -f$ honest processes received an incorrect prediction about $p_j$.
\label{misclass-honest}
\end{observation}

\begin{lemma}
\label{Flemma:advice_preprocess_misclassify_bound}
    If $f < \epsilon n$ for some constant $0 < \epsilon < 1/2$, then there are at most $O(B/n)$ misclassified processes.

\end{lemma}
\begin{proof}
By Observation \ref{misclass-faulty},
for each faulty process $p_j$ that was misclassified, there were at least $\lceil (n+1)/2 \rceil -f$ honest processes that received an incorrect prediction about $p_j$.
Furthermore, by Observation  \ref{misclass-honest},
for each misclassified honest process $p_j$,  there were at least $\lceil n/2 \rceil -f$ honest processes that received an incorrect prediction about $p_j$. So, for each misclassified process, $p_j$, there were at least $\lceil n/2 \rceil -f$ honest processes that received an incorrect prediction about $p_j$. 
Thus,  there are at most $B/ (\lceil n/2 \rceil -f) \leq B/((1/2 - \epsilon)n) \in O(B/n)$ misclassified processes.
\end{proof}

We use $k_H$ to denote the actual number of honest processes that are misclassified. Similarly, we use $k_F$ to denote the actual number of faulty processes that are misclassified. Finally, we use $k_A = k_F+k_H$ to denote the actual number of misclassified processes.

For any classification vector $c \in \{0,1\}^n$,
let $\pi(c)$ denote the ordering of the process identifiers $\{1,\ldots, n\}$
where the identifiers of the processes classified as honest are in increasing order and precede  
the identifiers of the processes classified as faulty, which are also in increasing order.
More formally,\\
\hspace{1in}$\bullet$ if $c[i] = 1$, then
$i$ is in position $\sum_{j=1}^i c[j]$ of $\pi(c)$ and\\
\hspace{1in}$\bullet$ if $c[i] = 0$, then
$i$ is in position $n-f + i-\sum_{j=1}^ic[j] = i+ \sum_{j=i+1}^n c[j]$ of $\pi(c)$.

Let $\hat{c}$ denote the correct classification of the processes, so $\hat{c}[i] = 1$ for all $i \in H$ and $\hat{c}[i] = 0$ for all $i \in F$.
Then, in $\pi(\hat{c})$, the integers in $H$ are in increasing order and precede the integers in $F$,
which are also  in increasing order.
If $i \in H$, then the position of $i$ in $\pi(\hat{c})$ is the number of honest processes with identifiers less than or equal to $i$,
which is $\sum_{j=1}^i \hat{c}[j]$.
If $i \in F$, then the position of $i$ in $\pi(\hat{c})$ is $n-f$, the number of honest processes,  plus the number of faulty processes with identifiers less than or equal to $i$. Since the number of faulty processes with identifiers less than $i$ is $i - \sum_{j=1}^i \hat{c}[j]$
and $n-f = \sum_{j=1}^n \hat{c}[j]$, it follows that the position of $i$ in $\pi(\hat{c})$ is
$n-f + i - \sum_{j=1}^i\hat{c}[j] = i+ \sum_{j=i+1}^n \hat{c}[j]$

Let $\delta(v,v')$ denote the Hamming distance between two vectors
$c,c' \in \{0,1\}^n$. This is the number of components in which their bits differ.
In particular, if $c$ misclassifies $m$ processes, then $\delta(c,\hat{c}) =m$.


\begin{lemma}
Let $c$ be a classification vector that misclassifies $m$ processes. 
If $c$ properly classifies the  process $p_i$, then the positions of $i$ in $\pi(c)$ and $\pi(\hat{c})$
differ by at most  $m$.
\label{lemma:differpos}
\end{lemma}

\begin{proof}
If $i \in H$, then $i$ is in position $\sum_{j=1}^i c[j]$ of $\pi(c)$ and in position $\sum_{j=1}^i \hat{c}[j]$ of $\pi(\hat{c})$.
Since $c$ and $\hat{c}$ differ in at most $m$ positions, $\sum_{j=1}^ic[j]$ and $\sum_{j=1}^i\hat{c}[j]$ differ by at most $m$.
If $i \in F$, then $i$ is in position $i+ \sum_{j=i+1}^n c[j]$ of $\pi(c)$ and in position $i+ \sum_{j=i+1}^n \hat{c}[j]$ of $\pi(\hat{c})$.
Since $c$ and $\hat{c}$ differ in at most $m$ positions, $\sum_{j=i+1}^n c[j]$ and $\sum_{j=i+1}^n \hat{c}[j]$ differ by at most $m$.
\end{proof}


\begin{lemma}
Let $c$ be a classification vector that misclassifies $m \leq n-f$ honest processes. 
If $c$ properly classifies the faulty process $p_i$, then the position of $i$ in 
$\pi(c)$ is greater than $n-f -m$.
\label{lemma:faultypos}
\end{lemma}

\begin{proof}
Since there are $m$ honest processes that are misclassified by $c$, 
there are at least  $n-f -m$ processes that are classified as honest by $c$.
The indices of all these processes precede $i$ in $\pi(c)$.
\end{proof}

If $c$ is the classification vector (of some honest process), then the number of
honest
processes that are misclassified by $c$ is at most $k_A$.
Since $f \leq t$, the contrapositive of Lemma~\ref{lemma:faultypos} gives the following result.


\begin{corollary}
If $c$ is a classification vector,  $p_i$ is a faulty process, and 
the position of $i$ in $\pi(c)$ is at most $n-t-k_A$, then $p_i$ is misclassified by $c$.
\label{corollary:faultypos}
\end{corollary}


\begin{lemma}
Suppose that $c$ and $c'$ are classification vectors
that both misclassify the faulty process $p_i$ as honest.
If there are $k_A$ misclassified processes, then
then the positions of $i$ in $\pi(c)$ and $\pi(c')$
differ by at most  $k_A-1$.
\label{lemma:faultydiffer}
\end{lemma}

\begin{proof}
Note that $i$ is in position $\sum_{j=1}^i c[j]$ of $\pi(c)$ and in position $\sum_{j=1}^i c'[j]$ of $\pi(c')$.
Since $p_i$ is misclassified by both processes, there are at most $k_A-1$ indices $j < i$ such that $p_j$
is misclassified by at least one of these classification vectors. 
Hence there are at most $k_A-1$ indices $j < i$ such that  $c[j] \neq c'[j]$ and, thus, the sums differ by at most $k_A-1$.
\end{proof}


Next, we show that, for any range of consecutive positions whose size, $s$, is larger than the number of misclassified processes, $k_A$,
and whose right endpoint is at most $n-t-k_A$, there is a set of at least $s-k_A$ honest processes 
who indices appear in that range of positions in the classification vector of every honest process.



\begin{lemma}
Suppose $ \ell + k_A- 1 < r \leq n-t-k_A$. Then there is a set $G \subseteq H$ of size at least $r-\ell +1 - k_A$ such that 
$G \subseteq \{ \pi(c_i)[j] \ | \ \ell \leq j \leq r \}$ for all $i \in H$.
\label{lemma:coresize}
\end{lemma}

\begin{proof}
Let $\hat{c}$ be the correct classification vector and
let $I$ denote the set of identifiers of processes that are properly classified by all honest processes
and occur in positions $\ell$ through $r$ of $\pi(\hat{c})$.
Since $r \leq n-t$, these are all identifiers of honest processes, so $I \subseteq H$.

Let $h' = r - \ell + 1 - |I|$ be the number of honest processes that are misclassified by at least one
honest process and whose identifiers occur in $\ell$ through $r$ of $\pi(\hat{c})$.
Then, each of these identifiers is among the last $\min\{t+k, n\}$ positions of $\pi(c_i)$ for at least one $i \in H$. 

Let $h$ be the number of misclassified honest processes whose identifiers occur among the first $\ell -1$ positions of $\pi(\hat{c})$.
For each $i \in H$, 
there is some $h'' \leq h$ such that the $h''$ smallest identifiers in $I$ occur among the first $\ell -1$
positions of $\pi(c_i)$.
This is a result of processes being misclassified as faulty by process $p_i$, so
their identifiers occur among the last $\min\{t+k, n\}$ positions in $\pi(c_i)$, rather than among the first $\ell -1$ positions.

Let $b$ be the number of faulty processes with identifiers less than $\pi(\hat{c})[r]$ that are misclassified by at least one
honest process.
For each $i \in H$, there is some $b'' \leq b$ such that the $b''$ largest identifiers in $I$
do not occur among the first $r$ positions of $\pi(c_i)$.
This is a result of faulty processes being misclassified as honest by process $p_i$, so
their identifiers occur among the first $r$ positions in $\pi(c_i)$, rather than among the 
last $t$ positions.

The identifiers of faulty processes that are misclassified as honest by $c_i$, 
but are at least $\pi(\hat{c})[r]$, appear in $\pi(c_i)$ after the
identifiers in $I$, so they do not affect the position of identifiers in $I$.
Similarly, the identifiers of honest processes that are misclassified as faulty by $c_i$, but are at least $\pi(\hat{c})[r]$, do not affect the position of identifiers in $I$.

Let $G \subseteq I$ consist of all identifiers in $I$,
except for the $h$ smallest identifiers and the $b$ largest identifiers.
Then $G \subseteq \{ \pi(c_i)[j] \ | \ \ell \leq j \leq r \}$ for all $i \in H$.
Furthermore, since $h' + h + b \leq k_A$ (as they count different types of misclassified processes) and $I$ has size $r-\ell+1-h'$, it follows that $G$ has size
at least $r-\ell +1 -k$.
\end{proof}

We call the set $G$ from \Cref{lemma:coresize} a \emph{core set}.


\begin{lemma}
If $ r \leq n-t-k_H$, then the number of honest processes $p_i$ such that $i \in \{ \pi(c_i)[j] \ |\ 1 \leq j \leq r\}$ is at most $r+k_H$.
\label{lemma:honest-broadcasters}
\end{lemma}

\begin{proof}
Let $h \le k_H$ be the number of honest processes  with identifiers less than or equal to  $\pi(\hat{c})[r+k_H]$
that are misclassified as faulty by at least one honest process.
At most $h$ of these processes are misclassified by $p_i$, causing their identifiers to occur later in $\pi(c_i)$ than in $\pi(\hat{c})$ and causing at most $h$ of the identifiers occurring immediately after position $r$ in $\pi(\hat{c})$ to be in the first $r$ positions  of $\pi(c_i)$. In other words,
besides the $r$ identifiers in $\{ \pi(\hat{c})[j] \ | \ 1 \leq j \leq r \}$, the only other identifiers $i$
for which it is possible that $i \in \{ \pi(c_i)[j] \ |\ 1 \leq j \leq r\}$ are the $h \leq k_H$ identifiers
in $\{ \pi(\hat{c})[j] \ | \ r+1 \leq j \leq r+h \}$.
\end{proof}


\begin{ignore}
\begin{lemma}
Suppose $ \ell \leq r \leq n-t-k_A$. Then
$D = \bigcup_{i \in H} \{ \pi(c_i)[j] \ | \ \ell \leq j \leq r \}$ has size at most $r-\ell + 1 + 2k_A$.
\label{L9}
\end{lemma}

\begin{proof}
Let $h$ be the number of honest processes  with identifiers less than or equal to  $\pi(\hat{c})[r+k_A]$
that are misclassified as faulty by at least one honest process.
At most $h$ of these processes are misclassified by $p_i$, causing their identifiers to occur later in $\pi(c_i)$ than in $\pi(\hat{c})$ and causing at most $h$ of the identifiers occurring immediately after position $r$ in $\pi(\hat{c})$ to be in positions $\ell$ through $r$ of $\pi(c_i)$. In other words,
for each $i \in H$, at most $h$ identifiers in $\{ \pi(\hat{c})[j] \ |\ r < j \leq r+h\}$ occur in positions $\ell$ through $r$ of $\pi(c_i)$.

Let $b$ be the number of faulty processes with identifiers less than $\pi(\hat{c})[\ell]$ that are misclassified as honest by at least one honest process.
At most $b$ of these processes are misclassified by $p_i$, causing their identifiers to occur earlier in $\pi(c_i)$ than in $\pi(\hat{c})$ and causing at most $b$ of the identifiers occurring immediately before position $\ell$ in $\pi(\hat{c})$ to be in positions $\ell$ through $r$ of $\pi(c_i)$. In other words,
dor each $i \in H$, at most $b$ identifiers in $\{ \pi(\hat{c})[j] \ |\ \ell -b \leq j < \ell\}$ occur in positions
$\ell$ through $r$ of $\pi(c_i)$.

In addition, a faulty process that is misclassified as honest by process $p_i$ may occur in positions
$\ell$ through $r$ of $\pi(c_i)$. 

There are $r - \ell +1$ identifiers in $\{ \pi(\hat{c})[j] \ | \ \ell \leq j \leq r \}$.
Since $h + b \leq k_H + k_F = k_A$, it follows that $D$ contains at most $k_A + k_F \leq 2k_A$ other elements.
\end{proof}

\end{ignore}

%% file: faith-unauthenticated.tex
\section{Unauthenticated Byzantine Agreement using Classification}\label{section:un-byz-adv}

This section presents a conditional unauthenticated Byzantine agreement algorithm that makes use of the classification described in the previous section. It relies on the condition that all honest processes have the same upper bound $k$ on the number of misclassified processes, where
$(2k+1)(3k+1)+k \leq n-t$.  Interestingly, the algorithm does not require $t < n/3$, which is necessary for unauthenticated Byzantine agreement without predictions
\cite{FLM85}.
If the number of misclassified processes is larger than $k$, it is possible that the honest processes do not decide.
However, if any honest process decides, then they all decide, although not necessarily in the same round.


The algorithm consists of $2k+1$ phases, each consisting of 5 rounds.
Each honest process broadcasts in at most one phase of the algorithm. This keeps the message complexity from being too large.
In each phase, 
each honest process has a set of processes that it listens to. These sets can be different for different processes, but will have a large set of honest processes in common.
During a phase, the honest processes first use a variant of graded consensus to determine whether there is already enough agreement among their current values. Then they use conciliation, which will result in the honest processes all having the same current value, provided 
they all had the same current value after the graded consensus or  
each of them
is only listening to honest processes. Finally, they use graded consensus again to determine whether they can decide.

We begin with our graded consensus algorithm in \Cref{subsection:un-ac}, followed by the conciliation algorithm in
\Cref{subsection:concil}. The entire algorithm is presented in \Cref{subsection:BAalg-class} together with a proof of correctness.

\subsection{Graded Consensus with Core Set}\label{subsection:un-ac}


Our graded consensus algorithm is closely related to the graded consensus algorithm in \cite{civit2024efficient}.
There, each honest process, $p_i$, has an input value, $v_i$ (from some domain of allowed values).
Here, $p_i$ is also given the value $k$ and a subset $L_i\subseteq \{1,\ldots,n\}$ of $3k+1$ identifiers, which are the identifiers of the processes that $p_i$ listens to.

The algorithm consists of two rounds. In the first round, if $i \in \mathit{L}_i$, process $p_i$ broadcasts its input value, $v_i$.
If there is some value that a process $p_i$ has received from at least $2k+1$ different processes with identifiers in $L_i$,
then it sets its local variable $b_i$ to this value. Note that it ignores any messages it receives from processes that are not in $L_i$.
If $i \in \mathit{L}_i$ and $p_i$ set $b_i$ to an input value, then it broadcasts $b_i$ in the second round.
Consider a process $p_i$ that set $b_i$ to an input value in round 1.
If it receives $b_i$ from at least $2k+1$ different processes with identifiers in $L_i$ during the second round,
it returns the pair $(b_i,1)$. Otherwise, it returns the pair $(b_i,0)$.
Now consider a process that did not set $b_i$ during round 1.
If it receives the same value, $v'$ from at least $k+1$ different processes with identifiers in $L_i$ during the second round,
it returns the pair $(v',0)$. Otherwise, it returns the pair $(v_i,0)$.
In this algorithm, each honest processes $p_i$ such that $i \in \mathit{L}_i$ sends at most $2n$ messages. 
The remaining honest processes send no messages.

Suppose there is a core set $G\subseteq H$ of size at least $2k+1$ such that $G \subseteq L_i$ for all $i \in H$.
(Note that $G$ is not used in the algorithm and is not known to the processes.)
Then we prove that
the values returned by honest processes
satisfy the following two
properties:
\begin{compactitem}
    \item \emph{Strong Unanimity}: 
If every honest process has the same input value, $v$, then they all return $(v, 1)$.
    \item \emph{Coherence}: If an honest process returns $(v, 1)$, then every honest process returns the value $v$.
\end{compactitem}

\begin{algorithm} [h]
\caption{Unauthenticated Graded Consensus with Core Set: (code for process $p_i$)}
\label{algorithm:un_gc_core_set}
\footnotesize
\ \\
\textbf{Conditions under which strong unanimity and coherence are guaranteed:}\\
$|L_i| = 3k +1$ for all $i \in H$\\
there exists $G \subseteq H$ with $|G| \geq 2k+1$ such that $G \subseteq L_i$ for all $i \in H$\\
\ \\
\textbf{graded-consensus-with-core-set($v_i,k,L_i$):}
\begin{algorithmic} [1] 
\State \textbf{Round 1}:
\State \hskip2em \textbf{if} $i \in \mathit{L}_i$ \textbf{then broadcast} $v_i$ \label{line:bcast1}
\State \hskip2em Let $R_i$  denote the multiset of values received from processes with identifiers in $\mathit{L}_i$\label{line:received1}
\State \hskip2em \textbf{if} there is a value $v$ which occurs at least $2k+1$ times in $R_i$ 
\State \hskip2em \textbf{then} $b_i \leftarrow v$ \label{line:setb}
\State \hskip2em \textbf{else} $b_i \leftarrow \bot$

\smallskip
\State \textbf{Round 2}:
    \State \hskip2em \textbf{if} $i \in \mathit{L}_i$ and $b_i \ne \bot$ \textbf{then broadcast} $b_i$ \label{line:bcast2}
\State \hskip2em Let $R'_i$  denote the multiset of values received from processes with identifiers  in $\mathit{L}_i$\label{line:received2}
\State \hskip2em \textbf{if} $b_i \neq \bot$ 
\State \hskip2em \textbf{then if} $b_i$ occurs at least $2k+1$ times in $R'_i$ 
\State \hskip4em \textbf{then} \textbf{return}($b_i,1$) \label{line:retb1}
\State \hskip4em \textbf{else} \textbf{return}($b_i,0$) \label{line:retb0}
\State \hskip2em \textbf{else if} there is a value $v'$ which occurs at least $k+1$ times in $R'_i$ 
\State \hskip4em \textbf{then} \textbf{return}($v',0$)\label{line:retv}
\State \hskip4em \textbf{else} \textbf{return}($v_i,0$)\label{line:retinput}
\end{algorithmic} 
\end{algorithm}

\begin{lemma}
\label{lemma:gc-b}
All honest processes that perform \Cref{line:setb} 
have the same value for $v$. 
\end{lemma}
\begin{proof}
Consider any honest process, $p_i$, that performs \Cref{line:setb}.
It follows from the code that the value $b_i\neq \bot$ occurred at least $2k+1$ times in $R_i$.
Since $|L_i - G| \leq k$, process $p_i$ must have received $b_i$ from at least $k+1$
different processes with identifiers in $G$. Since honest processes in $G$ broadcast their inputs
on \Cref{line:bcast1}, it follows that $b_i$ is the input of an honest process.

Consider any other honest process, $p_j$. Since $G \subseteq L_j$ and all processes with identifiers in $G$ are honest, it also received $b_i$ from at least $k+1$ different processes.
But $L_j$ has size $3k+1$, so $p_j$ received at most $2k$ values that are different than $b_i$.
Hence, if $p_j$ performs \Cref{line:setb}, then $b_j = b_i$.
\end{proof}

\begin{ignore}
\begin{lemma}
\label{lemma:gc-r}
Each honest process that performs \Cref{line:retv} returns the value $b_i\neq \bot$ of some  
honest process $p_i$.
\end{lemma}

\begin{proof}
During the second round, if a process $p_j$ received the value $v'$ from at least $k+1$ different processes, it received $v'$ from
some process $p_i$ whose identifier is in $G$, since $|L_i - G| \leq k$. 
This honest process broadcast $v' = b_i$ on \Cref{line:bcast2}.
\end{proof}

\begin{lemma}[Validity]
    \label{lemma:validity}
The value returned by any honest process is the input of some honest process.
\end{lemma}
\begin{proof}
If an honest process returns on \Cref{line:retb1} or \Cref{line:retb0}, then the claim
follows from \Cref{lemma:gc-b}.
If it returns on \Cref{line:retv}, then the claim follows from \Cref{lemma:gc-r}.
Otherwise, it returns its own input on \Cref{line:retinput}.
\end{proof}
\end{ignore}

\begin{lemma}[Strong Unanimity]
    \label{lemma:unanimity}
If every honest process has the same input value, $v$, then they all return $(v, 1)$.\end{lemma}
\begin{proof}
Suppose that $v_i = v$ for all $i \in H$.
Since $G \subseteq H$ and $G \subseteq L_i$ for all $i \in H$, it follows that $i \in L_i$ for all $i \in G$. Thus, at least $2k+1$ processes with identifiers in $G$ broadcast $v$ on \Cref{line:bcast1}
and each honest process, $p_i$, has at least $2k+1$ copies of $v$ in $R_i$, so it sets $b_i =v$.
During round 2, each process whose identifier is in $G$ broadcasts $v$ on \Cref{line:bcast2}
and each honest process, $p_i$, has at least $2k+1$ copies of $v$ in $R'_i$, so it returns
$(v,1)$ on \Cref{line:retb1}.
\end{proof}

\begin{lemma}[Coherence]
    \label{lemma:coherence}
If an honest process returns $(v, 1)$, then every honest process returns the value $v$.
\end{lemma}
\begin{proof}
Let $p_i$ be an honest process that returns $(v,1)$. Then, from the code,
$v= b_i \neq \bot$ and the value $v$ occurred in $R'_i$ at least $2k+1$ times.
Since $|L_i - G| \leq k$, process $p_i$ must have received $v$ from at least $k+1$
different processes with identifiers in $G$ during round 2.

Consider any other honest process $p_j$.
If $b_j \neq \bot$, then, by \Cref{lemma:gc-b}, $b_j = b_i$. Hence, either it returns $(v,1)$ on \Cref{line:retb1} or it returns $(v,0)$ on \Cref{line:retb0}.

So, suppose that $b_j = \bot$.
Since $v = b_i \neq \bot$, the value $v$ occurred in $R'_i$ at least $2k+1$ times. 
Since $G \subseteq L_j$ 
and process $p_i$ received $v$ from at least $k+1$ different processes with identifiers in $G$ during round 2, so did $p_j$.
Thus $v$ occurs at least $k+1$ times in $R'_j$. 
By \Cref{lemma:gc-b}, all processes in $G$ that perform 
\Cref{line:setb} use the value $v$, so they broadcast $v$ on \Cref{line:bcast2}.
Since $|L_j - G| \leq k$, no other value occurred in $R'_j$ more than $k$ times.
Thus $p_j$ returns $(v, 0)$ on \Cref{line:retv}.
\end{proof}

\subsection{Conciliation with Core Set}\label{subsection:concil}

In the conciliation algorithm,
each honest process, $p_i$, has an input value, $v_i$. It is also given the value $k$ and a subset $L_i\subseteq \{1,\ldots,n\}$ of $3k+1$ identifiers.

The algorithm 
consists of a single round in which each process, $p_i$, with $i \in L_i$, broadcasts its input value $v_i$ and its set $L_i$.
Each process, $p_i$, constructs a directed graph whose vertex set, $T_i$, consists of the identifiers of the processes from which it received a message this round and whose edges are the ordered pairs $(y,z)$ of distinct vertices in $T_i$ such that $y \in L_z$.
For each vertex $z \in T_i \cap L_i$,  process $p_i$ computes the minimum, $m_i[z]$, of the input values, $v_y$, of processes $p_y$ for which there is a path from $y$ to $z$ in this graph. Finally, $p_i$ returns a value that occurs
most frequently among these values.

Suppose that $L_i$ only contains identifiers of honest processes
and there is a core set $G$ of size at least $2k+1$ such that $G \subseteq L_i$, for every honest process, $p_i$.
Then we prove that the values returned by honest processes
satisfy the following two properties:
\begin{compactitem}
    \item \emph{Agreement}: All honest processes return the same value.
    \item \emph{Strong Unanimity}: 
If every honest process has the same input value, $v$, they all return $v$.
\end{compactitem}

\begin{algorithm} [h]
\caption{Conciliation with Core Set: (code for process $p_i$)}
\label{algo:un_concil_core_set}
\footnotesize
\ \\
\textbf{Conditions under which agreement and strong unanimity are guaranteed:}\\
$|L_i| = 3k +1$ for all $i \in H$\\
$L_i \subseteq H$ for all $i \in H$\\
There exists $G \subseteq H$ with $|G| \geq 2k+1$ such that $G \subseteq L_i$ for all $i \in H$\\
\ \\
\textbf{conciliate}($v_i,k,L_i$):
\begin{algorithmic} [1] 
\State \hskip2em \textbf{if} $p_i \in \mathit{L}_i$ \textbf{then broadcast} $v_i,L_i$ \label{line:bcast}
\State \hskip2em Let $T_i$ be the identifiers of the processes from which $p_i$ received a message.
\State \hskip2em Let $E_i = \{ (y,z)\ |\ y,z \in T_i \mbox{ and } y \in L_z\}$
\State \hskip2em {\bf for each} $z \in L_i$, let $m_i[z] = \min \{ v_y \ |\ y \in L_y$ and there exists a path from $y$ to $z$ in the graph $(T_i,E_i)$\}
\State \hskip2em Let $v'_i$ denote a value that occurs the largest number of times in the multiset $\{m_i[j]\ |\ j \in T_i \cap L_i\}$
\State \hskip2em \textbf{return}$(v'_i)$
\end{algorithmic} 
\end{algorithm}

Let $T = \{i \in H \ |\ i \in L_i\}$
denote the set of identifiers of honest processes that broadcast during the algorithm.
Then $G \subseteq T \subseteq T_i$ for all $i \in H$ and every process in $T_i - T$ is faulty.

\begin{lemma}\label{lemma:path}
If 
there is a path in the graph $(T_i,E_i)$ from a vertex $y \in T_i$ to
a vertex $j \in T$, then $y \in T$.
\end{lemma}

\begin{proof}
Consider any edge $(y,z) \in E_i$ such that $z \in T$.
Since $z \in T \subseteq H$, it follows that $y \in L_z \subseteq H$. Hence $y \in T$.
The claim follows by induction on the length of the path.
\end{proof}

\begin{lemma}\label{lemma:mvalues1}
If $j,j' \in G$ and  $i \in H$,
then $m_i[j] = m_i[j']$.
\end{lemma}

\begin{proof}
Since $j,j' \in G \subseteq L_j \cap L_{j'}$, both $(j,j')$ and $(j',j)$ are in $E_i$.
Hence, in the graph $(T_i, E_i)$, there is a path from $y \in T_i$ to $j$ if and only if there is a path from $y$ to $j'$.
It follows that $m_i[j] = m_i[j']$.
\end{proof}

\begin{lemma}\label{lemma:mvalues2}
If 
$j \in G$ and $i,i' \in H$,
then $m_i[j] = m_{i'}[j]$ for all $j \in G$ and all $i,i' \in H$.
\end{lemma}

\begin{proof}
Since $j \in G \subseteq T$, \Cref{lemma:path} implies that every path from $y$ to $j$ in the graph $(T_i,E_i)$
only consists of vertices in $T$. Thus, this path also exists in the graph $(T_{i'},E_{i'})$.
Hence $m_{i'}[j] \leq m_i[j]$. By symmetry, $m_i[j] \leq m_{i'}[j]$, so $m_i[j] = m_{i'}[j]$.
\end{proof}

\begin{lemma}[Agreement]\label{lemma:concil-agree}
All honest processes return the same value.
\end{lemma}

\begin{proof}
By \Cref{lemma:mvalues1}, each honest process, $p_i$, computes the same value,
$m_i[j]$ for each $j \in G$. Since $|G| \geq 2k+1$ and $|T_i \cap L_i| \leq |L_i| = 3k+1$,
this is the value $v'_i$ that $p_i$ returns.
By \Cref{lemma:mvalues2}, $m_i[j] = m_{i'}[j]$ for all $j \in G$ and $i,i' \in H$.
Thus, every two honest processes, $p_i$ and $p_{i'}$, return the same value.
\end{proof}

\begin{lemma}[Strong Unanimity]\label{lemma:concil-unan}
If every honest process 
has the same input value $v$, then
they all return $v$.
\end{lemma}

\begin{proof}
By \Cref{lemma:path}, if there is a path from $y$ to $j\in G \subseteq T$ in $(T_i,E_i)$, then 
$y \in T \subseteq H$. Since $v_y = v$, it follows that $m_i[j] = v$.
Since $|G| \ge 2k+1$, they all return $v$.
\end{proof}

\subsection{Byzantine Agreement Algorithm with Classification}\label{subsection:BAalg-class}


\Cref{algo:un_ba_class} is similar to the early-stopping phase-king
algorithm by Lenzen and Sheikholeslami~\cite{lenzen2022}.
However, in addition, each process is given a value $k$ and the classification vector, $c_i$,
obtained from the predictions, as described in \Cref{sec:classification}.

The first $(2k+1)(3k+1)$ identifiers in the sequence $\pi(c_i)$
are partitioned into $2k+1$ consecutive blocks of size $3k+1$.
The algorithm consists of at most $2k+1$ phases.
In phase $\phi$, process $p_i$ uses the $\phi$'th block of identifiers as its set $L_i$. 

In each phase, $p_i$ calls graded consensus, followed
by conciliation, and then a separate instance of graded consensus.
It starts with $v_i$ as its input value and updates it to the value
returned from graded consensus.
If the grade returned from the first call to graded consensus
is 0, it updates its current value, $v_i$, to be the value returned 
from conciliation.
If the grade returned from the second call to graded consensus is 1,
it decides its current value on 
\Cref{line:ba-decide}
and sets $decided_i$ to true.
If this is not phase $2k+1$, then $p_i$ participates in one additional phase
before it returns,
to ensure that all other honest
processes that have not yet decided will decide the same value.
Note that, even if its current value, $v_i$, changes during this additional phase
(because some of the other honest processes have already terminated),
it will not change the value of $decision_i$.
We will show that if the number of misclassified processes is at most $k$ and
honest process $p_i$ returns on \Cref{line:un-ba-ret2},
then, immediately before it performed this step, 
it decided value $v_i$.
However, if the number of misclassified processes is larger, process $p_i$ 
is not guaranteed to decide and the value it returns might not agree with
the values returned by other honest processes.

If all honest processes have the same input value $v$,
they all return $(v,1)$ from the first call to graded consensus.
In this case, the call to conciliation does not change their current values,
they all return $(v,1)$ from the second call to graded consensus, 
decide $v$, and terminate at the end of the second phase.

Otherwise, suppose that, in some phase $\phi \leq 2k$, all honest processes return the same value, $v$,
but possibly different grades,
from the second call to graded consensus.
Then, in phase $\phi +1$, they all return value $v$ and grade 1
from the first call to graded consensus.
As above, the call to conciliation does not change their current values
and they all return $(v,1)$ from the second call to graded consensus.
The honest processes that decided $v$ at the end of phase $\phi$ return.
The remaining honest processes decide $v$ at the end of phase $\phi+1$.
If $\phi = 2k$, then they also return at the end of phase $\phi+1$
If $\phi < 2k$, they return at the end of phase $\phi+2$.

To ensure that this case will eventually happen when $k$ is an upper bound on the number of
misclassified processes, we prove
that there is a phase, $\phi$, in which $L_i \subseteq H$ for all honest processes $p_i$.
In this phase, if their inputs to conciliation are all the same,
then they all return this value
from conciliation and they use this value as input to the second call to graded consensus.
If some of their inputs to conciliation are different, they all have grade 0, so they
will all set their current value to the value returned from conciliation.
In either case, all honest processes will return the same value and grade 1 from the
second call to graded consensus.

\begin{algorithm} [h]
\caption{Unauthenticated Byzantine Agreement with Classification: (code for process $p_i$)}
\label{algo:un_ba_class}
\footnotesize
\textbf{Conditions under which agreement and strong unanimity are guaranteed:}\\
$k$ is an upper bound on the number of misclassified processes\\
$(2k+1)(3k +1)\leq n-t-k$\\
$c_i =$ {\bf classify}($a_i$)\\
\ \\
\textbf{ba-with-classification($x_i,c_i,k$):}
\begin{algorithmic} [1] 
\State $v_i \leftarrow x_i$
\State $order_i \leftarrow \pi(c_i)$
\State $\mathit{decided}_i \gets \mathit{false}$
\State \textbf{for} $\phi \gets 1$ to $2k+1$:
\State \hskip2em Let $L_i = \{ order_i[j] \ |\ (3k+1)(\phi-1) +1 \leq j \leq (3k+1)\phi\}$
\State \hskip2em $(v_i,g_i) \gets$ \textbf{graded-consensus-with-core-set}($v_i,k,L_i$)\label{line:ba-gc1}
\State \hskip2em $v'_i \gets$ \textbf{conciliate}$(v_i,k,L_i)$\label{line:ba-concil}
\State \hskip2em \textbf{if} $g_i = 0$ \textbf{then} $v_i \gets v'_i$\label{line:ba-useconcil}
\State \hskip2em $(v_i,g_i) \gets$ \textbf{graded-consensus-with-core-set}($v_i,k,L_i$)\label{line:ba-gc2}
\State \hskip2em \textbf{if} $\mathit{decided}_i$ \textbf{then return}($decision_i$)\label{line:un-ba-ret1}
\State \hskip2em \textbf{if} $g_i = 1$ \textbf{then}\label{line:ba-test}
\State \hskip4em $decision_i \leftarrow v_i$\label{line:ba-decide}
\State \hskip4em $\mathit{decided}_i \gets \mathit{true}$
\State \textbf{return}($v_i$) \label{line:un-ba-ret2} 
\end{algorithmic} 
\end{algorithm}

Note that, for each phase $\phi$,
$\ell = (3k+1)(\phi-1) +1$,  $r = (3k+1)\phi \leq (2k+1)(3k+1) \leq n-t-k$, and $r- \ell + 1 = 3k+1$,
so $|L_i| = 3k+1$ for all honest processes $p_i$.
Furthermore, by \Cref{lemma:coresize}, there is
a set $G \subseteq H$ of size at least $r- \ell + 1 - k = 2k+1$ such that  $G \subseteq L_i$ for all $i \in H$.

\begin{lemma}
    \label{lemma:un-ba-two}
Let $p_j$ be a faulty process. Then, there are at most two (consecutive) phases in which 
$j \in \mathit{L}_i$ for some honest process $p_i$.
\end{lemma}
\begin{proof}
Suppose that $p_j$ is a faulty process such that $j \in L_i$ for some $i \in H$ during some phase $\phi$.
Then the position of $j$ in $order_i$ is between $\ell = (3k+1)(\phi-1)+1$ and  $r= (3k+1)\phi \leq
(3k+1)(2k+1) \leq n-t-k$.
By \Cref{corollary:faultypos}, $p_j$ is misclassified as honest by $p_i$. Furthermore, by \Cref{lemma:faultydiffer}, if $p_j$ is misclassified as honest by another honest process $p_{i'}$,
then the positions of $j$ in $order_i$ and $order_{i'}$ differ by at most $k-1$. 
If the position of $j$ in $order_i$ is between $\ell > (3k+1)$ and $\ell + k -2$,
then it is possible that $j \in L_{i'}$ during phase $\phi - 1$.
If the position of $j$ in $order_i$ is between $r-k+2$ and $r \leq (3k+1)2k$,
then it is possible that $j \in L_{i'}$ during phase $\phi + 1$.
Otherwise, $j \in L_{i'}$ during phase $\phi$.
Hence, there are at most two (consecutive) phases in which
$j \in \mathit{L}_i$ for some honest process $p_i$.
\end{proof}

\begin{lemma}[Strong Unanimity]
\label{lemma:un_ba_strong_unanimity}
If every honest process has the same input value, $v$, then they all return $v$.
\end{lemma}
\begin{proof}
If all honest processes have the same input value $v$, then in phase 1, by \Cref{lemma:unanimity},
they all return $(v,1)$ from the call to \textbf{graded-consensus-with-core-set} on \Cref{line:ba-gc1}.
In this case, the call to \textbf{conciliate} does not change their current values
on \Cref{line:ba-useconcil}. Since their current values are all $v$ when they 
call \textbf{graded-consensus-with-core-set} on \Cref{line:ba-gc2}, they all return $(v,1)$ by \Cref{lemma:unanimity}.
Hence, they all
decide $v$ on \Cref{line:ba-decide} and return on \Cref{line:un-ba-ret1} at the end of phase 2.
\end{proof}

\begin{lemma}
\label{lemma:un-decide}
If some honest process decides value $v$ in phase $\phi < 2k+1$,
then every honest process decides $v$ by the end of phase $\phi + 1$.
\end{lemma}
\begin{proof}
Suppose that the first phase in which any honest process decides (on \Cref{line:ba-decide})
is $\phi < 2k+1$. Let $p_i$ be an honest process that decided in phase $\phi$ and let $decision_i = v_i = v$ be the value it decided.
Then, by \Cref{line:ba-test}, $g_i = 1$. By \Cref{lemma:coherence},
every honest process, $p_j$, returned the value 
$v_j = v$ from \textbf{graded-consensus-with-core-set} on \Cref{line:ba-gc2} during phase $\phi$. Hence, if $p_j$ decided in phase $\phi$,
it decided value $v$.

Since an  honest process that decides prior to phase $2k+1$ returns one phase after it decides,
no honest process returns before phase $\phi+1$.
When the honest processes perform \textbf{graded-consensus-with-core-set} on \Cref{line:ba-gc1} during phase $\phi+1$,
they all have the same input value, $v$, so by \Cref{lemma:unanimity},
they all return $(v,1)$ from this call.
Thus, when they perform \Cref{line:ba-useconcil}, their current values remain unchanged, so
they all have the same input value, $v$, when they next perform \textbf{graded-consensus-with-core-set} on line \Cref{line:ba-gc2}.
By \Cref{lemma:unanimity}, they all return $(v,1)$ from this call, too.
Hence, if process $p_j$ did not decide in phase $\phi$, it decides $v$ on on \Cref{line:ba-decide} during phase $\phi+1$.
\end{proof}

\begin{lemma}\label{lemma:un-if-good-phase}
Suppose that $L_i \subseteq H$ for every honest process, $p_i$, during phase $\phi$.
Then all honest processes decide the same value by the end of phase $\phi$.
\end{lemma}

\begin{proof}
If some honest process decides $v$ in phase $\phi' < \phi$, then, by \Cref{lemma:un-decide},
they all decide $v$ by the end of phase $\phi'+1 \leq \phi$.
so, suppose that no honest process has decided prior to phase $\phi$.

First suppose that, during phase $\phi$, some honest process returns $(v,1)$ from the call to \textbf{graded-consensus-with-core-set} on \Cref{line:ba-gc1}. Then, by \Cref{lemma:coherence}, all honest processes return the value $v$ from this call. Hence, they all have input $v$ in the call to \textbf{conciliate} on \Cref{line:ba-concil}. By \Cref{lemma:concil-unan}, they all return the value $v$ from this call.
Hence, they all have input $v$ in the call to \textbf{graded-consensus-with-core-set} on \Cref{line:ba-gc2}.
By \Cref{lemma:unanimity}, they all return $(v,1)$ from this call. Hence, they all
decide $v$ on \Cref{line:ba-decide}.

Otherwise, during phase $\phi$, all honest processes return grade 0 from the call to \textbf{graded-consensus-with-core-set} on \Cref{line:ba-gc1}. By \Cref{lemma:concil-agree}, they all return the same value, $v$,
from the call to \textbf{conciliate} on \Cref{line:ba-concil}. Hence, each honest process, $p_i$, sets $v_i$ to $v$ on \Cref{line:ba-useconcil} and they all have input $v$ in the call to \textbf{graded-consensus-with-core-set} on \Cref{line:ba-gc2}.
By \Cref{lemma:unanimity}, they all return $(v,1)$ from this call. Hence, they all
decide $v$ on \Cref{line:ba-decide}.
\end{proof}

\begin{lemma}
\label{lemma:un-good-phase}
There is a phase $\phi$ during which $L_i \subseteq H$ for every honest process, $p_i$.
\end{lemma}

\begin{proof}
By \Cref{corollary:faultypos}, if $c$ is a classification vector (of an honest process), $p_j$ is a faulty process, and the position of $j$ in $\pi(c)$ is at most $n-t-k$, then $p_j$ is misclassified by $c$.
Since $(2k+1)(3k+1) \leq n-t-k$, 
if $j$ is among the first $(2k+1)(3k+1)$ positions of $c$, then $p_j$ was misclassified by $c$.
By \Cref{lemma:un-ba-two}, there are at most two phases in which $j \in L_i$ for some honest process $p_i$. There are $k_F \leq k$ misclassified faulty processes, so, in at least one of the $2k+1$ phases, $L_i \subseteq H$ for every honest process, $p_i$.
\end{proof}

\begin{lemma}[Agreement]
\label{lemma:un-agree}
All honest processes return the same value.
\end{lemma}

\begin{proof}
By \Cref{lemma:un-good-phase},
there is a phase $\phi$ during which $L_i \subseteq H$ for every honest process, $p_i$,
and, by \Cref{lemma:un-if-good-phase}, all honest processes decide the same value 
by the end of phase $\phi$. 
Since every honest process returns the value it decides, either in the phase after it decides
or in phase $2k+1$, they all return the same value.
\end{proof}

\begin{theorem}
\label{theorem:unauth_ba_class}
If $k$ is an upper bound on the number of misclassified processes in the classification vectors
of the honest processes and $(2k+1)(3k+1) \leq n-t-k$,
then \Cref{algo:un_ba_class} is a correct Byzantine agreement algorithm (i.e., it satisfies Agreement and Strong Unanimity) 
and a total of $O(nk^2)$ messages are sent by the honest processes.
Even if the number of misclassified processes is larger than $k$,
each honest process returns within $5(2k+1)$ rounds and sends at most $5n$ messages. 
\end{theorem}

\begin{proof}
Correctness follows from \Cref{lemma:un-agree} and \Cref{lemma:un_ba_strong_unanimity}.
From \Cref{line:un-ba-ret2}, every honest process returns by the end of phase $2k+1$.
Each call to {\bf graded-consensus-with-core-set} takes 2 rounds and each call to {\bf conciliate} takes 1 round,
so each phase takes 5 rounds. Thus every honest process returns within $5(2k+1)$ rounds.

Since 
the sets, $L_i$, used by an honest process, $p_i$, in different phases are disjoint,
there is at most one phase in which $i \in L_i$.
From the code of \Cref{algorithm:un_gc_core_set} and \Cref{algo:un_concil_core_set}, it
broadcasts at most once each round of a phase in which $i \in L_i$ and, otherwise, does not
send any messages. Thus it sends at most $5n$ messages.

If $k$ is an upper bound on the number of misclassified 
processes, then,
by \Cref{lemma:honest-broadcasters}, the number of honest processes $p_i$ such that $i \in L_i$ during one of
the $2k+1$ phases is at most $(2k+1)(3k+1) + k_H \leq (2k+1)(3k+1) + k$. Thus, the honest processes send a total of
$O(nk^2)$ messages.
\end{proof}



%% file: authenticated_v2.tex
\section{Authenticated Byzantine Agreement using Classification}\label{section:auth_byz_using_class}

This section presents a conditional authenticated Byzantine agreement algorithm that makes use of the classification described in \Cref{sec:classification}. Again, it is assumed that honest processes have the same upper bound $k$ on the number of misclassified processes. Unlike the unauthenticated algorithm, this algorithm only needs $2k+1 \le n-t-k$. However, 
the algorithm also needs that the upper bound, $t$, on the number of faulty processes to be less than $n/2$. We note that if the actual number of misclassified processes is larger than $k$, nothing is guaranteed 
about the outputs  
of honest processes.

In \Cref{subsection:crypto_tools},
we start by
describing
the cryptographic tools we use.
Next, in \Cref{subsection:bb_committee}, we present Byzantine Broadcast with Implicit Committee, which is 
an important component of our Byzantine Agreement algorithm.
The entire Byzantine Agreement algorithm and its analysis are 
presented in \Cref{subsection:auth_byz_class}

\subsection{Cryptographic Tools}\label{subsection:crypto_tools}

We assume the existence of a public-key infrastructure.
We use $\mathsf{sign}_i(m)$ to denote a signature created by process $p_i$ for the message $m$.
Each process can verify that this signature was created by $p_i$.
We assume that no (computationally-bounded) faulty process can forge this signature
if $p_i$ is honest.


A key part of our algorithm is to use a committee, a set of processes which will be listened to by all honest processes. Processes not in the committee will be ignored.
A process $p_i$ demonstrates that it is in the committee by attaching a \emph{committee certificate} for $p_i$ to its messages.


\begin{definition}[Committee Certificate]
    \label{definition:committee_certificate}
A committee certificate for process $p_i$ is a set of signatures for the 
message $\langle \textsc{committee}, p_i\rangle$
by $t+1$ different processes.
\end{definition}

Although some of the signatures in a committee certificate might be from faulty processes,
every committee certificate contains signatures from at least one honest process. In \Cref{subsection:auth_byz_class}, we will show how to use this fact and the classification obtained from \Cref{sec:classification} to obtain a committee which has size $O(k)$ and has at most $k$ faulty processes, effectively reducing the number of faulty processes (that honest processes listen) from at most $t$ to $k$.


A \emph{message chain} is another cryptographic object that is used.

\begin{definition}[Message Chain]
    \label{definition:message_chain}

A message chain of length 1 for value $x$  started by process $p_i$ is a message
$\langle x, \mathit{cc}_i, \mathsf{sign}_i(\langle x, \mathit{cc}_i \rangle)\rangle$, where
$\mathit{cc}_i$ is a committee certificate for $p_i$.
If $m$ is a message chain of length $b \ge 1$ for value $x$  started by process $p_i$
and $cc_j$ is a committee certificate for $p_j$,
then the message $\langle m, \mathit{cc}_j, \mathsf{sign}_j(\langle m, \mathit{cc}_j \rangle)\rangle$
is a message chain of length $b+1$ for value $x$ started by process $p_i$.
A message chain of length $b \ge 1$ is valid if it is signed by $b$ different processes.
\end{definition}

If there are at most $k$ faulty processes in the committee, a message chain of length $k+1$ must contain a committee certificate for some honest process $p_j$ and a signature by $p_j$. In \Cref{subsection:bb_committee}, we use this to ensure that when an honest process receives a message chain of length $k+1$ for value $v$ started by process $p_s$,
every other honest process must have received 
a (possibly shorter) message chain for value $v$ started by process $p_s$.

\subsection{Byzantine Broadcast with an Implicit Committee}\label{subsection:bb_committee}



\Cref{algo:bb_implicit_committee}
is very similar to the Byzantine Broadcast algorithm by Dolev and Strong \cite{dolev1983authenticated}.
In their algorithm,
there is a designated sender,  which is supposed to broadcast its input value. 
Each process receives the identifier, $s$, of the sender as an input parameter.
Here, each honest process $p_i$ also receives $k$ and possibly a committee certificate
for $p_i$.

Crucially, $p_i$ only listens to messages from a process $p_j$ if the message contains a committee certificate for $p_j$. We say the committee is implicit because honest processes do not know who the members of the committees are. However, each member, $p_j$, of the committee can prove its membership in the committee using the committee certificate for $p_j$.
We say the process $p_j$ \emph{has a committee certificate} if $cc_j$  is a committee certificate for $p_j$.



The algorithm consists of $k+1$ rounds. During the algorithm,
when an honest process $p_i$ receives a valid message chain for some value $x$ started by $p_s$, it records $x$ in the set $X_i$, provided that $X_i$ contains fewer than 2 values.
In the first round, if the sender has a committee certificate, it broadcasts a message chain of length $1$ for its input value, $x_s$.
In each of the $k$ subsequent rounds, if process $p_i$ has a committee certificate
and it added the value $x$ to $X_i$ at the end of the previous round,
then it extends the length of a valid message chain for $x$ started by $p_s$
(that it received in the previous round) and broadcasts it.
At the end of round $k+1$, 
if $X_i$ contains exactly one element, $x$, then $p_i$ returns the value $x$. Otherwise, it returns $\bot$.


Let $C$ denote the set of identifiers of processes that have committee certificates.
Note that $C$ is not used in the algorithm and is not known by the honest processes.
If $|C \cap F| \le k$, we prove that the values returned by honest processes satisfy the following three properties:
\begin{compactitem}
    \item \emph{Committee Agreement}: All honest processes that have a committee certificate return the same value.
    \item \emph{Validity with Sender Certificate}: If the sender, $p_s$, is honest and $cc_s$ is a committee certificate for $p_s$, 
    then every honest process returns its input value, $x_s$.
    \item \emph{Default without Sender Certificate}: If the sender, $p_s$, does not have a committee certificate, then every honest process returns $\bot$.
\end{compactitem}

\begin{ignore}
\Faith{Are these standard names for these properties? If not, we should think about making them more meaningful.
How about "Validity with Certificate" or "Validity with Sender Certificate"
instead of "Committee Validity" and
"Default without Certificate" or "Default without Sender Certificate" instead of "Non-Committee Validity"?
The "Committee" in "Committee Agreement" refers to the processes that are returning values,
whereas the "Committee" in "Committee Validity" refers to the sending being in the committee.}
\Ayaz{For standard Byzantine broadcast, usually it is

\begin{itemize}
    \item Agreement: All honest processes return the same value
    \item Validity: if the sender $p_s$ is honest, then every honest process returns its input value $x_s$.
\end{itemize}

I preprended Committee to tell that ``this is only true for processes with committee certificate (in Agreement) or if the sender has a committee certificate (in Validity).

Non-Committee Validity, I wanted to say ``this is what happens if a non-committee becomes the sender''

anyway, I think these do not have standard names yet. at least I am not aware of any prior definitions for this kind of Byzantine broadcast.}

\begin{algorithm} [h]
\caption{Byzantine Broadcast with Implicit Committee: (code for process $p_i$)}
\label{algo:bb_implicit_committee}
\footnotesize
\textbf{Conditions under which committee agreement, committee validity, and non-committee validity are guaranteed:}\\
There are at most $k$ faulty process $p_j$ such that a committee certificate for $p_j$ exists\\
\ \\
\textbf{bb-with-implicit-committee($sender_i, x_i, k, cc_i$):} \BlueComment{$x_i$ is only used if $p_i = sender_i$}
\begin{algorithmic} [1]
\State \textbf{Round 1}:
    \State \hskip2em $X_i \gets \emptyset$
    \State \hskip2em \textbf{if} $sender_i = p_i$ and $cc_i \ne \bot$ \textbf{then}    
        \State \hskip4em $X_i \gets \{x_i\}$
        \State \hskip4em \textbf{broadcast} $\langle x_i, cc_i, \mathsf{sign}_i(\langle x_i, cc_i\rangle)\rangle$

\State \textbf{Round $j \in [2, k+1]$}:
    \State \hskip2em \textbf{for each} valid message chain $m$ of length $j-1$ started by $sender_i$ that is received from round $j-1$:
        \State \hskip4em Let $x_m$ be the value corresponding to $m$
        \State \hskip4em \textbf{if} $|X_i| < 2$ and $x_m \notin X_i$ \textbf{then}\label{line:bb_implicit_committee_check_xm}
            \State \hskip6em $X_i \gets X_i \cup \{x_m\}$
            \State \hskip6em \textbf{if} $cc_i \ne \bot$ \textbf{then} \textbf{broadcast} $\langle m, cc_i, \mathsf{sign}_i(\langle m, cc_i\rangle)\rangle$
\State \textbf{Round $k+1$}: \BlueComment{at the end of the round}
    \State \hskip2em \textbf{for each} valid message chain $m$ of length $k+1$ started by $sender_i$ that is received from round $k+1$:
        \State \hskip4em Let $x_m$ be the value corresponding to $m$
        \State \hskip4em \textbf{if} $|X_i| < 2$ and $x_m \notin X_i$ \textbf{then} $X_i \gets X_i \cup \{x_m\}$
    \State \hskip2em \textbf{if} $X_i = \{x\}$ \textbf{then} \textbf{return}($x$)
    \State \hskip2em \textbf{else} \textbf{return}($\bot$)
\end{algorithmic} 
\end{algorithm}
\end{ignore}

\begin{algorithm} [h]
\caption{Byzantine Broadcast with Implicit Committee: (code for process $p_i$)}
\label{algo:bb_implicit_committee}
\footnotesize
\textbf{Conditions under which the committee agreement, validity with sender certificate, and default without sender certificate properties are guaranteed:}\\
There are at most $k$ faulty processes $p_j$ such that $cc_j$ is a committee certificate for $p_j$\\
\ \\
\textbf{bb-with-implicit-committee($s, x_i, k, cc_i$):}
\begin{algorithmic} [1]
    \State $X_i \gets \emptyset$
    \State\textbf{if} $s = i$ and $cc_i$ is a committee certificate for $p_i$ \textbf{then}    
        \State \hskip2em $X_i \gets \{x_i\}$\label{line:addx1}
        \State \hskip2em \textbf{broadcast} $\langle x_i, cc_i, \mathsf{sign}_i(\langle x_i, cc_i\rangle)\rangle$\label{line:start-chain}
    \State $R_i \gets$ the set of valid message chains of length 1  started by $p_s$ that $p_i$ received in round 1  
\smallskip
\State \textbf{for} $j \leftarrow 2$ to $k+1$:
    \State \hskip2em\textbf{for each} $m \in R_i$: 
        \State \hskip4em \textbf{if} $m$ is a message chain for value $x$, $x \not\in X_i$, and $|X_i| < 2$ \textbf{then} 
            \State \hskip6em $X_i \gets X_i \cup \{x\}$\label{line:addx2}
            \State \hskip6em \textbf{if} $cc_i$ is a committee certificate for $p_i$ \textbf{then} \textbf{broadcast} $\langle m, cc_i, \mathsf{sign}_i(\langle m, cc_i\rangle)\rangle$
    \State \hskip2em $R_i \gets$ the set of valid message chains of length $j$ started by $p_s$ that $p_i$ received in round $j$
\smallskip
\State \textbf{for each} $m \in R_i$: 
    \State \hskip2em \textbf{if} $m$ is a message chain for value $x$, $x \not\in X_i$, and $|X_i| < 2$ \textbf{then} $X_i \gets X_i \cup \{x_m\}$\label{line:addx3}
    \State \textbf{if} $X_i = \{x\}$ 
    \State \hskip2em \textbf{then} \textbf{return}($x$)
    \State \hskip2em \textbf{else} \textbf{return}($\bot$)
\end{algorithmic} 
\end{algorithm}


\begin{lemma}[Validity with Sender Certificate]
    \label{lemma:bb_with_implicit_committee_committee_validity}
If the sender, $p_s$, is honest and $cc_s$ is a committee certificate for $p_s$, then every honest process returns its input value, $x_s$.
\end{lemma}
\begin{proof}
Only one message chain is broadcast by $p_s$ on \Cref{line:start-chain} and it is for value $x_s$.
In round 1, 
every other honest process receives this message chain for $x_s$ from $p_s$.
Thus, every honest process, $p_i$, adds $x_s$ to $X_i$ in round 1.

By the unforgeability of signatures, every message chain started by $p_s$
is for the value $x_s$. 
Thus, no honest process, $p_i$ adds any other value to $X_i$ during the algorithm.
Hence, at the end of the algorithm, every honest process returns the value $x_s$.
\end{proof}

\begin{lemma}[Default without Sender Certificate]
    \label{lemma:bb_with_implicit_committee_non_committee_validity}
If the sender, $p_s$, does not have a committee certificate, then every honest process returns $\bot$.
\end{lemma}
\begin{proof}
If $p_s$ does not have a committee certificate,
then no message chain can be started by $p_s$.
Hence, at the end of the algorithm, $X_i = \emptyset$ for each honest process, $p_i$,
so it returns $\bot$.
\end{proof}

\begin{lemma}[Committee Agreement]
\label{lemma:bb_with_implicit_committee_committee_agreement}
If $|F \cap C| \leq k$, then
all honest processes that have a committee certificate return the same value.
\end{lemma}
\begin{proof}
Suppose there are two honest processes, $p_i$ and $p_j$, that both have a committee certificate.
Consider any value $x \in X_i$.
If $i = s$, then $p_s$ started a valid message chain for value $x_s \in X_s$ and broadcast it
in round 1. This is the only valid message chain it started, so it does not add any other elements
to $X_s$.
If $p_i$ received a valid message chain for value $x$ started by $p_s$
during one of the first $k$ rounds,
then $p_i$ broadcast a valid message chain for value $x$ started by $p_s$.
If $p_i$ only received a valid message chain $m$ for value $x$ started by $p_s$
during round $k+1$, then $m$ has length $k+1$, so it was signed by $k+1$ different processes,
each of which has a committee certificate. Since there are at most $k$ faulty processes that have
committee certificates, at least one of these $k+1$ processes was honest.
That honest process broadcast a valid message chain for value $x$ started by $p_s$.
In all cases, process $p_j$ received a valid message chain for value $x$.
Hence, at the end of the algorithm, either $x \in X_j$ or $|X_j| \geq 2$.

Consider the sets $X_i$ and $X_j$ at the end of the algorithm.
If $|X_i| \geq 2$, then $|X_j| \geq 2$ and both $p_i$ and $p_j$ return $\bot$.
Similarly, $|X_j| \geq 2$ implies that $|X_i| \geq 2$.
So, suppose that $|X_i|\leq 1$ and $|X_j| \leq 1$.
If $X_i= \{x\}$, then $x \in X_j$, so $X_j = \{x\}$  and both $p_i$ and $p_j$ return $x$.
Similarly, $|X_j| =1$ implies that $|X_i| = 1$.
The remaining case is when $X_i$ and $X_j$ are both empty, in which case,
both $p_i$ and $p_j$ return $\bot$.
Hence, in all cases, processes $p_i$ and $p_j$ return the same value.
\begin{ignore}
    \textbf{Case 1:} $output_i, output_j \ne \bot$. At the end of the algorithm, $X_i = \{output_i\}$ and $X_j = \{output_j\}$.
    If $p_i$ is the sender or $p_i$ obtains $output_i$ from a message chain of length at most $k$, then $p_i$ must have sent a message chain of length $k+1$ for $output_i$ to $p_j$. Furthermore, if $p_i$ obtains $output_i$ from a message chain of length $k+1$, from the precondition, there must be a committee certificate of an honest process in the chain, which must have sent a (possibly shorter) message chain for the same value to $p_j$ as well. In any case, $output_i = output_j$, a contradiction.

    \textbf{Case 2:} $output_i = \bot$ or $output_j = \bot$. Without loss of generality, suppose $output_i = \bot$ and $output_j \ne \bot$. Hence, at the end of the algorithm $X_j = \{output_j\}$ and $|X_i| = 0$ or $|X_i| = 2$. As shown in the previous case, it is impossible to have $|X_i| = 0$. However, since $|X_i| = 2$, using similar argument on each $x \in X_i$, it is impossible to have $|X_j| = 1$. A contradiction.
\end{ignore}

Therefore, all honest processes that have a committee certificate return the same value.
\end{proof}

\begin{observation}
If the sender, $p_s$, has a committee certificate, then each honest process that has a committee
certificate broadcasts at most twice and honest processes that do not have a committee certificate
do not send any messages. If the sender does not have a committee certificate, then no honest process
sends any messages.
\end{observation}

\begin{ignore}
Lastly, we prove the complexities. We can see that even if there are more than $k$ faulty processes with a committee certificate, the number of rounds and messages will not change.

\begin{lemma}
    \label{lemma:bb_implicit_committee_complexities}
     Let $p_s$ be the designated sender. Each honest process returns from \Cref{algo:bb_implicit_committee} in $k+1$ rounds. If $p_s$ has a committee certificate then honest processes send $O(|C|n)$ messages, and no messages otherwise.
\end{lemma}
\begin{proof}
    The number of rounds is explicit in the algorithm. If $p_s$ does not have a committee certificate, no message chain can be started by $p_s$ and thus, no message will be sent. Otherwise, the number of messages follows from the fact that each honest process with a committee certificate broadcasts at most twice (see the check at Line~\ref{line:bb_implicit_committee_check_xm}).
\end{proof}
\end{ignore}

\subsection{Byzantine Agreement Algorithm with Classification}\label{subsection:auth_byz_class}

\Cref{algo:auth_ba_class} is a Byzantine Agreement Algorithm that uses classification
and takes advantage of authentication.
It follows a well-known reduction from Byzantine Agreement to $n$ instances of Byzantine
Broadcast, one with each process as the sender, that works for $t < n/2$ faulty processes.
In addition to its input value, an honest process $p_i$ is also given a value $k$, which is an upper bound on the number of misclassified processes, and a classification vector $c_i$ obtained from the predictions as described in \Cref{sec:classification}. 

In a nutshell, 
the algorithm 
makes use of the orderings obtained from the classifications to limit the power of the adversary.
Each honest process $p_i$ will only
send a signature for $\langle \textsc{committee}, p_j\rangle$ to process $p_j$ if $j$ is among the first $2k+1$ identifiers in the ordering
$\pi(c_i)$ obtained from its classification vector.
An honest process $p_j$ that receives at least $t+1$ signatures from different processes
for $\langle \textsc{committee}, p_j\rangle$ has a committee certificate, $cc_j$.
We will show that, provided $2k+1 < n-t-k$,
at most $k$ faulty processes have a committee certificate
and at least $k+1$ honest processes have a committee certificate.
By restricting an honest process to only listen to processes that have a committee certificate,
the number of faulty processes that it listens to is bounded above by $k$, instead of by $t$.

In the next $k+1$ rounds, each process participates in $n$ parallel instances of Byzantine Broadcast with Implicit Committee (\Cref{algo:bb_implicit_committee}), where $p_s$ is the designated sender in instance $s$. 
They are all given $k$ as an upper bound on the number of faulty processes that have a committee certificate.
Note that, in instance $s$, only message chains started by $p_s$ are (supposed to be) broadcast,
so each process can identify which instance each message belongs to.

In the final round, each (honest) process with a committee certificate  broadcasts the smallest value other than $\bot$ that occurred the largest
number of times among their outputs from the $n$ instances of Byzantine Broadcast, together with its committee certificate.
We will show that all honest processes that have a committee certificate output the same value from each instance of Byzantine Broadcast,
so they broadcast the same value, $v$, in this round.
Among the multiset of values an honest process receives in such messages in this round, it returns the smallest value that occurred the largest number of times.
Since there are more honest processes that have a committee certificate than faulty processes that have a committee certificate, all honest processes
return the value $v$.
When all honest processes (that have a committee certificate) have the same input value, then $v$ is this value.



\begin{algorithm} [h]
\caption{Authenticated Byzantine Agreement with Classification: (code for process $p_i$)}
\label{algo:auth_ba_class}
\footnotesize
\textbf{Conditions under which agreement and strong unanimity are guaranteed}\\
$k$ is an upper bound on the number of misclassified processes\\
$2k+1 \le n-t-k$\\
$t < n/2$\\
$c_i = \text{classify}(a_i)$\\
\ \\
\textbf{ba-with-classification($x_i, c_i, k$):}
\begin{algorithmic} [1] 
\State $order_i \gets \pi(c_i)$
\State Let $L_i = \{ order_i[j] \ |\ 1 \le j \le 2k+1\}$
\State \textbf{for each} $j \in L_i$: \textbf{send} $\mathsf{sign}_i(\langle \textsc{committee}, p_j\rangle)$ to $p_j$\label{line:auth_ba_send_sign}

\State $cc_i \gets \emptyset$
\State Let $S_i$ be the identifiers of the processes $p_j$ that sent $\mathsf{sign}_j(\langle \textsc{committee}, p_i\rangle)$ to $p_i$.
\State \textbf{if} $|S_i| \geq t+1$ \textbf{then} $cc_i \gets \{\mathsf{sign}_j(\langle \textsc{committee}, p_i\rangle) \ |\ j$ is one of the $t+1$ smallest identifiers in $S_i\}$ 
\smallskip

\State \textbf{for each} $s \in [1, n]$, \textbf{in parallel}:
    \State \hskip2em $bb_i[s] \gets $ \textbf{bb-with-implicit-committee}($s, v_i, k, cc_i$) 

\smallskip
\State \textbf{if} $cc_i$ is a committee certificate for $p_i$ \textbf{then}  
    \State \hskip2em Let $plurality_i$ be the smallest value other than $\bot$ that occurs the largest number of times in $bb_i$\label{line:auth_ba_compute_bb_plural}
    \State \hskip2em \textbf{broadcast} $\langle plurality_i, cc_i \rangle$\label{line:auth_ba_broadcast_bb_plural}
\State Let $V_i$ be the multiset of $plurality_j$ values received from processes that have a committee certificate\label{line:auth_ba_Vi}
\State Let $decision_i$ be the smallest value that occurs the largest number of times in $V_i$
\State \textbf{return}($decision_i$)
\end{algorithmic} 
\end{algorithm}


As in the previous section, let $C$ denote the set of identifiers of processes that have committee certificates. 

\begin{lemma}
    \label{lemma:auth_ba_committee_stats}
If $2k+1 \le n-t-k$, then    $|C| \le 3k+1$, $|C \cap F| \le k$, and $|C \cap H| \ge k+1$.
\end{lemma}
\begin{proof}
If process $p_j$ has a committee certificate, then there is an honest process $p_i$ with $j \in L_i$. 
Since $L_i$ consists of the first $2k+1 \leq n-t -k$ identifiers from $order_i$,
\Cref{corollary:faultypos} implies that every faulty process in $C$ is a misclassified faulty process. Hence $|C \cap F| \le k_F \leq k$.
By \Cref{lemma:honest-broadcasters}, $|C \cap H| \leq (2k+1) + k_H$, so $|C| = |C \cap H| + |C\cap F| \leq 2k+ 1 + k_H + k_F \leq 3k+1$.
Finally, if $j \in L_i$ for each $i \in H$, then $p_j$ has a committee certificate, as $n-t \ge t+1$. Hence, from \Cref{lemma:coresize}, $|H \cap C| \ge 2k+1 - k = k+1$.
\end{proof}


\begin{lemma}
    \label{lemma:auth_ba_bb_plural}
All honest processes that have a committee certificate broadcast the same value on \Cref{line:auth_ba_broadcast_bb_plural}.
\end{lemma}
\begin{proof}
From 
    \Cref{lemma:bb_with_implicit_committee_committee_agreement} and the second part of \Cref{lemma:auth_ba_committee_stats},
all honest processes that have a committee certificate return the same value from each instance of  Byzantine Broadcast.   
Each of these processes, $p_i$, deterministically computes the value,  $plurality_i$, that it will broadcast 
on \Cref{line:auth_ba_broadcast_bb_plural}
from the array, $bb_i$, containing these return values.
Thus, they all broadcast the same value.
\end{proof}

\begin{ignore}
\begin{lemma}[Termination]
    Every honest process decides.
\end{lemma}
\begin{proof}
    Each honest process decides at the end of round $k+3$.
\end{proof}
\end{ignore}

\begin{lemma}[Agreement]
    \label{lemma:auth_ba_agreement}
    All honest processes return the same value.
\end{lemma}
\begin{proof}
Consider any honest process, $p_i$.
Each value in $V_i$ was sent by a process in $C$.
By \Cref{lemma:auth_ba_committee_stats}, $V_i$ contains at least $k+1$ values broadcast by honest processes
and at most $k$ values sent by faulty processes. By \Cref{lemma:auth_ba_bb_plural}, all honest processes
broadcast the same value, so this value occurs in $V_i$ more than than any other value.
Hence, process $p_i$ returns this value.
\end{proof}

\begin{lemma}[Strong Unanimity]
    If every honest process (in $C$) has the same input value, $v$, then every honest process returns $v$.\label{lemma:auth_ba_unan}
\end{lemma}
\begin{proof}
Consider any honest process $p_i$.
By \Cref{lemma:bb_with_implicit_committee_non_committee_validity}, $bb_i[j] = \bot$ for each $j \notin C$
and, by \Cref{lemma:bb_with_implicit_committee_committee_validity}, $bb_i[j] = v$ for each $j \in C \cap H$.
By \Cref{lemma:auth_ba_committee_stats}, $|C\cap H| > |C\cap F|$, so $plurality_j = v$ for each $j \in C \cap H$.
Thus, $V_i$ contains more copies of $v$ than any value other than $\bot$, so $p_i$ returns $v$.
\end{proof}

\begin{ignore}
We now show the complexities of the algorithm. We also care about the complexities when there are more than $k$ misclassified processes (\Cref{lemma:auth_ba_complexities_abnormal}) as our final algorithm may call this algorithm under that condition.

\begin{lemma}
    \label{lemma:auth_ba_complexities_normal}
    Each honest process returns from \Cref{algo:auth_ba_class} in $k+3$ rounds and the total number of messages sent by honest processes is $O(nk^2)$.
\end{lemma}
\begin{proof}
    The number of rounds is explicit from the algorithm. To count the number of messages, recall that $|C| \le 4k+1$ from \Cref{lemma:auth_ba_committee_stats}. Besides the Byzantine Broadcasts, honest processes send $O(nk + |C|n) = O(nk)$ at Line~\ref{line:auth_ba_send_sign} and Line~\ref{line:auth_ba_broadcast_bb_plural}. Then, from \Cref{lemma:bb_implicit_committee_complexities}, honest processes send $O(|C|^2n) = O(nk^2)$ messages in the Byzantine Broadcasts. In total, honest processes send $O(nk^2)$ messages.
\end{proof}

\begin{lemma}
    \label{lemma:auth_ba_complexities_abnormal}
    If there are more than $k$ misclassified processes, each honest process returns from \Cref{algo:auth_ba_class} in $k+3$ rounds and the total number of messages sent by honest processes is $O(n^3)$.
\end{lemma}
\begin{proof}
    The number of rounds is explicit from the algorithm and not affected by how many misclassified processes there are. To count the number of messages, by definition, $|C| \le n$. Using the arguments from \Cref{lemma:auth_ba_complexities_normal}, honest processes send $O(n^3)$ messages.
\end{proof}

We conclude this section with the following theorem.
\end{ignore}

\begin{theorem}
    \label{theorem:auth_ba_class}
If $k$ is an upper bound on the number of misclassified processes in the classification vectors
of the honest processes, $2k+1 \leq n-t-k$, and $t < n/2$,
then \Cref{algo:un_ba_class} is a correct Byzantine agreement algorithm (i.e., it satisfies Agreement and Strong Unanimity)
and a total of $O(nk^2)$ messages are sent by the honest processes.
Even if the number of misclassified processes is larger than $k$,
each honest process returns after $k+3$ rounds and 
sends $O(n^2)$ messages.
\end{theorem}
\begin{proof}
Correctness follows from \Cref{lemma:auth_ba_agreement} and \Cref{lemma:auth_ba_unan}.
Each honest process returns after $k+3$ rounds, since it returns 
from \textbf{bb-with-implicit-committee} (\Cref{algo:bb_implicit_committee})
after $k+1$ 
rounds.

Each honest process sends $2k+1$ messages on \Cref{line:auth_ba_send_sign}.
If $i \in C$, then $p_i$ broadcasts one message on \Cref{line:auth_ba_broadcast_bb_plural}.
During each instance of \textbf{bb-with-implicit-committee}($s, v_i, k, cc_i$)  with $s \not\in C$,
there are no valid message chains started by $p_s$, so $p_i$ does not send any messages.
During each instance of \textbf{bb-with-implicit-committee}($s, v_i, k, cc_i$) with $s \in C$,
process $p_i$ broadcasts at most two messages (each time it adds a value to $X_i$ during the first $k$ rounds).
Thus, $p_i$ sends at most $2k+1 + n + 2n|C| < 2n(|C| + 1) \in O(n^2)$ messages.

If $2k+1 \leq n-t-k$, then $|C| \le 3k+1$, by \Cref{lemma:auth_ba_committee_stats}, so
the a total of $O(nk^2)$ messages are sent by the honest processes.
\end{proof}

%% file: final_algo_proofs.tex
\section{Correctness of Byzantine Agreement with Predictions}\label{section:final_algo_proof}


This section proves the correctness of \Cref{algo:final_ba_copy}
and then analyzes its round and message complexities.

\setcounter{algorithm}{0}
\begin{algorithm} [h]
\caption{Byzantine Agreement with Predictions: (code for process $p_i$)}
\label{algo:final_ba_copy}
\footnotesize
\textbf{ba-with-predictions($x_i,a_i$):}
\begin{algorithmic} [1] 
\State $c_i \leftarrow \textbf{classify}(a_i)$
\State $v_i \leftarrow x_i$
\State $\mathit{decided}_i \gets \mathit{false}$
\State \textbf{for} $\phi \gets 1$ to $\lceil\log_2t\rceil + 1$:
    \State \hskip2em $T \gets \alpha \cdot 2^{\phi-1}$       
    \BlueComment{Choose $\alpha$ so \textbf{ba-early-stopping} and ba-classification can complete for this phase.}
    
    \State \hskip2em $(v_i,g_i) \gets$ \textbf{graded-consensus}($v_i$)\label{line:final-ba-copy-gc1}  
    \BlueComment{Graded Consensus protects agreement.}
    \State \hskip2em $v'_i \gets$ \textbf{ba-early-stopping}($v_i, t, T$)\label{line:final-ba-copy-early-stopping} \BlueComment{early-stopping BA with time limit of $T$ rounds}
    \State \hskip2em \textbf{if} $g_i = 0$ \textbf{then} $v_i \gets v'_i$\label{line:final-ba-copy-use-early-stoppping}
    
    \State \hskip2em $(v_i,g_i) \gets$ \textbf{graded-consensus}($v_i$)\label{line:final-ba-copy-gc2} 
    \State \hskip2em $v'_i \gets$ \textbf{ba-with-classification}($v_i, c_i, 2^{\phi-1}, T$)\label{line:final-ba-copy-classification} 
    \BlueComment{at most $2^{\phi-1}$ classification errors, time limit $T$}
    \State \hskip2em \textbf{if} $g_i = 0$ \textbf{then} $v_i \gets v'_i$\label{line:final-ba-copy-use-classification}

    \State \hskip2em $(v_i,g_i) \gets$ \textbf{graded-consensus}($v_i$)\label{line:final-ba-copy-gc-last}
    \BlueComment{Graded Consensus protects agreement.}
    \State \hskip2em \textbf{if} $\mathit{decided}_i$ \textbf{then return}($decision_i$)\label{line:final-ba-copy-normal-decide}
    \State \hskip2em \textbf{if} $g_i = 1$ \textbf{then}
        \State \hskip4em $decision_i \gets v_i$\label{line:final-ba-copy-decision}
        \State \hskip4em $\mathit{decided}_i \gets \mathit{true}$\label{line:final-ba-copy-decided}
\State \textbf{return}($decision_i$)\label{line:final-ba-copy-force-decide}
\end{algorithmic} 
\end{algorithm}

To prove correctness, we use the following facts. First, the outputs of \textbf{graded-consensus} satisfy Strong Unanimity and Coherence as long as there are at most $t$ faulty processes. Second, 
if all honest processes return from the original \textbf{ba-early-stopping} algorithm within $T$ rounds 
when there are $f \le t$ faulty processes, the outputs of our \textbf{ba-early-stopping} algorithm (that terminates within $T$ rounds)
satisfy Strong Unanimity and Agreement. 
We emphasize that the correctness of 
\textbf{ba-with-classification} is not needed to prove the correctness of \textbf{ba-with-prediction}---only the performance.




\begin{lemma}[Strong Unanimity]
    \label{lemma:final_ba_strong_unanimity}
    If every honest process has the same value, $v$, at the beginning of phase $\phi$, then they all return $v$
    by the end of phase $\min \{ \phi + 1, \lceil \log_2 t \rceil + 1 \}$.
    In particular, if every honest process has the same input value, they all return this value.
\end{lemma}
\begin{proof}
Suppose every honest process has the same value, $v$, at the beginning of phase $\phi$.
Then they input $v$ to \textbf{graded-consensus} on \Cref{line:final-ba-copy-gc1}. 
By the Strong Unanimity property of \textbf{graded-consensus},    
all honest processes return $(v, 1)$ and thus, each honest process, $p_i$, sets $v_i$ to $v$ and $g_i$ to $1$. Moreover, as $g_i = 1$, each honest process ignores the output of \textbf{ba-early-stopping} on \Cref{line:final-ba-copy-early-stopping}.
Likewise, all honest processes return $(v, 1)$ on \Cref{line:final-ba-copy-gc2}, ignore the output of \textbf{ba-with-classification} on \Cref{line:final-ba-copy-classification}, and return $(v, 1)$ on \Cref{line:final-ba-copy-gc-last}.
If $decided_i = true$, then process $p_i$ had set $decision_i$ to $v_i = v$ at the end of the previous phase, so it returns $v$
on \Cref{line:final-ba-copy-normal-decide}.
If not, honest process $p_i$ sets $decision_i$ to $v$ on \Cref{line:final-ba-copy-decision} and $decided_i$ to $true$ \Cref{line:final-ba-copy-decided}. If $\phi = \lceil \log_2 t \rceil+1$, then it returns $v$ on \Cref{line:final-ba-copy-force-decide}. Otherwise, it returns $v$ on \Cref{line:final-ba-copy-normal-decide} in phase $\phi+1$.
\end{proof}

\begin{lemma}
    \label{lemma:final_ba_decision_before_last_phase}
Suppose there is an honest process, $p_i$, that sets  $decision_i$ to $v$ in phase $\phi \le \lceil \log_2 t \rceil$. Then every honest process returns $v$ by the end of phase $\min\{\phi+2, \lceil \log_2 t \rceil + 1\}$.
\end{lemma}
\begin{proof}
Let $p_i$ be the first honest process that performs \Cref{line:final-ba-copy-decision} and suppose it sets
$decision_i$ to $v$. Then $p_i$ returned $(v,1)$ when it last performed \textbf{graded-consensus} on \Cref{line:final-ba-copy-gc-last}.
By the Coherence property of \textbf{graded-consensus}, every honest process, $p_j$, returned $(v_j,g_j)$, with $v_j =v$ on \Cref{line:final-ba-copy-gc-last}.
Thus, at the beginning of phase $\phi+1$, every honest process has the same $v$.
By \Cref{lemma:final_ba_strong_unanimity}, every honest process returns $v$ by the end of phase $\min\{\phi+2, \lceil \log_2 t \rceil + 1\}$. 
\end{proof}


\begin{lemma}
    \label{lemma:final-ba-conditional-decide-early-ba}
If the preconditions of \textbf{ba-early-stopping} hold in phase $\phi$ (every honest process output from \textbf{ba-early-stopping} within $T = \alpha 2^{\phi-1}$ rounds), then all honest processes return the same value 
by the end of phase $\min\{\phi +1, \lceil \log_2 t \rceil + 1\}$.
\end{lemma}
\begin{proof}
If $\phi > 0$ and some honest process $p_i$ set $decision_i$ to $v$ in a previous phase, then, by \Cref{lemma:final_ba_decision_before_last_phase},
all honest processes return $v$ by the end of phase $\min\{\phi +1, \lceil \log_2 t \rceil + 1\}$.
So, suppose no honest process has performed \Cref{line:final-ba-copy-decision} prior to phase $\phi$.
Consider the output, $(v_i,g_i)$, of each honest process, $p_i$, from \textbf{graded-consensus} on \Cref{line:final-ba-copy-gc1} in phase $\phi$.
If $g_i = 0$ for every honest process, $p_i$, then on \Cref{line:final-ba-copy-use-early-stoppping}, they all set $v_i$ to $v'_i$, where
$v'_i$ is the value it returns from \textbf{ba-early-stopping} on 
\Cref{line:final-ba-copy-early-stopping}.
Since the preconditions of \textbf{ba-early-stopping} hold, the Agreement property says that these return values are all the same.
Otherwise, there is an honest process, $p_j$, with $g_j = 1$. From the Coherence property of \textbf{graded-consensus}, all honest processes returned
value $v_j$ from \textbf{graded-consensus} and, hence, they all used $v_j$ as input to \textbf{ba-early-stopping} on
\Cref{line:final-ba-copy-early-stopping}.
Since the preconditions of \textbf{ba-early-stopping} hold, the Strong Unanimity property says that all honest processes, $p_i$, return $v'_i = v_j$.
Thus, even if $g_i = 0$ and $p_i$ performs \Cref{line:final-ba-copy-use-early-stoppping}, $v_i$ still has value $v_j$.
In either case, all honest processes have the same input, $v$, to \textbf{graded-consensus} on \Cref{line:final-ba-copy-gc2}.

By Strong Unanimity, all honest processes $p_i$ return $(v,1)$, they ignore the output of \textbf{ba-with-classification} on \Cref{line:final-ba-copy-classification}, return $(v, 1)$ on \Cref{line:final-ba-copy-gc-last}, and set $decision_i$ to $v$ on \ref{line:final-ba-copy-decision}.
If $\phi = \lceil \log_2 t \rceil+1 $, they return $v$ on \Cref{line:final-ba-copy-force-decide}.
Otherwise, they return $v$ on \Cref{line:final-ba-copy-normal-decide} in phase $\phi+1$.
\end{proof}

\begin{ignore}
\begin{lemma}
    \label{lemma:final_ba_last_phase_no_decision}
    Suppose no honest process $p_i$ sets its $decision_i$ by the end of phase $\lceil \log_2 t \rceil$. Then, by the end of phase $\lceil \log_2 t \rceil + 1$, each honest process decides $v'$ for some value $v'$.
\end{lemma}
\begin{proof}
    At phase $\lceil \log_2 t \rceil + 1$, $T = ct$ for some $c$ such that when there are at most $t$ faulty processes, all honest processes return from \textbf{ba-early-stopping} in at most $T$ rounds and the output satisfies Strong Unanimity and Agreement. Thus, the lemma statement follows from \Cref{lemma:final-ba-conditional-decide-early-ba}. 
\end{proof}
\end{ignore}

\begin{lemma}[Agreement]
All honest processes decide the same value.
\end{lemma}
\begin{proof}
If there is some honest process, $p_i$, that sets $decision_i$ before phase $\lceil \log_2 t \rceil + 1$, Agreement follows from \Cref{lemma:final_ba_decision_before_last_phase}.
Otherwise, in phase $\lceil \log_2 t \rceil + 1$, the preconditions for \textbf{ba-early-stopping} hold, since $T= \alpha t$ rounds are sufficient for all honest processes to output from \textbf{ba-early-stopping} when there are at most $t$ faulty processes.
Thus, Agreement follows from \Cref{lemma:final-ba-conditional-decide-early-ba}. 
\end{proof}

Now, we focus on the round and message complexities of \Cref{algo:final_ba_copy}.

\begin{lemma}
    \label{lemma:final-ba-conditional-decide-classification-ba}
Suppose that, in phase $\phi$, $k = 2^{\phi -1}$ is an upper bound on the number of misclassified processes.
Let $T = \alpha 2^{\phi-1}$.
If $(2k+1)(3k+1) \leq n-t -k$  and $T\geq 5(2k+1)$ in the unauthenticated case
or $2k+1 \leq n-t-k$ and $T \geq k+3$ in the authenticated case,
then all honest processes
return by the end of phase $\min\{\phi +1, \lceil \log_2 t \rceil + 1\}$.
\end{lemma}
\begin{proof}
    If $\phi > 0$ and by the end of phase $\phi-1$ there is an honest process $p_i$ which 
 performed \Cref{line:final-ba-copy-decision}. Then  \Cref{lemma:final_ba_decision_before_last_phase}
 implies that every honest process returns by the end of phase $\min\{\phi+1, \lceil \log_2 t \rceil + 1\}$.

So, suppose that, at the start of phase $\phi$, no honest process has performed \Cref{line:final-ba-copy-decision}.
Consider the output, $(v_i,g_i)$, of each honest process, $p_i$, from \textbf{graded-consensus} on \Cref{line:final-ba-copy-gc2} in phase $\phi$.
If $g_i = 0$ for every honest process, $p_i$, then on \Cref{line:final-ba-copy-use-classification}, they all set $v_i$ to $v'_i$, where
$v'_i$ is the value it returns from \textbf{ba-with-classification} on 
\Cref{line:final-ba-copy-classification}.
By the Agreement property of \textbf{ba-with-classification}, they all return the same value.
Otherwise, there is an honest process, $p_j$, with $g_j = 1$. From the Coherence property of \textbf{graded-consensus}, all honest processes returned value $v_j$ from \textbf{graded-consensus} and, hence, they all used $v_j$ as input to \textbf{ba-with-classification}.
Then, by the Strong Unanimity property of \textbf{ba-with-classification}, they all return $v_j$.
Thus, even if $g_i = 0$ and $p_i$ performs \Cref{line:final-ba-copy-use-classification}, $v_i$ still has value $v_j$.
In either case, all honest processes have the same input to \textbf{graded-consensus} on \Cref{line:final-ba-copy-gc-last}.
By the Strong Unanimity property of \textbf{graded-consensus},
every honest process will return a grade of 1 and, hence,
return in phase $\min\{\phi+1, \lceil \log_2 t\rceil + 1\}$.
\end{proof}

\begin{lemma}
    \label{lemma:final-ba-complexities}
    Assume \textbf{graded-consensus} has $O(1)$ round complexity. Let $mc_{gc}$ be the message complexity of the \textbf{graded-consensus}, $mc_{early}$ be the message complexity of \textbf{ba-early-stopping}, and $mc_{class}$ be the message complexity of \textbf{ba-with-classification} even when their preconditions may not hold. If all honest processes return by the end of phase $\phi$, then the algorithm runs for $O(2^{\phi})$ rounds and, in total, honest processes send $O(n^2 + \phi\cdot(mc_{gc} + mc_{early} + mc_{class}))$ messages.
\end{lemma}
\begin{proof}
In \textbf{classify} (\Cref{algo:vote}),  honest processes perform one round and each broadcasts 1 message. 
In phase $\phi' \leq \phi$, honest processes spend $\alpha \cdot \min\{2^{\phi'-1}, t\}$ rounds for some constant $\alpha$ and send $3mc_{gc} + mc_{early} + mc_{class} \in O(mc_{gc} + mc_{early} + mc_{class})$ messages. 
As $\phi \le \lceil \log_2{t} \rceil + 1$, $\sum_{\phi' = 1}^{\phi} \alpha \cdot \min\{2^{\phi'-1}, t\} \in O(2^{\phi})$. Therefore, the algorithm runs for $O(2^{\phi})$ rounds and, in total, honest processes send $O(n^2 + \phi \cdot(mc_{gc} + mc_{early} + mc_{class}))$ messages.
\end{proof}

Before going through the main theorems, let us state several implementations of the early-stopping Byzantine Agreement and graded consensus that we use in proving the main theorems.

\begin{theorem}[Restated from \cite{civit2024efficient}, Unauthenticated Graded Consensus]
    \label{theorem:gc_n3}
    There is an unauthenticated algorithm that solves graded consensus when there are at most $t < n/3$ faulty processes. The algorithm runs for $2$ rounds and honest processes send $O(n^2)$ messages.
\end{theorem}

\begin{theorem}[Restated from \cite{momose2021}, Authenticated Graded Consensus]
    \label{theorem:gc_n2}
    There is an authenticated algorithm that solves graded consensus when there are at most $t < n/2$ faulty processes. The algorithm runs for $4$ rounds and honest processes send $O(n^2)$ messages.
\end{theorem}

\begin{theorem}[Restated from \cite{lenzen2022}, Unauthenticated Early-Stopping Byzantine Agreement]
    \label{theorem:early_stopping_ba_n3}
    There is an unauthenticated algorithm that solves Byzantine Agreement when there are at most $t < n/3$ faulty processes. The algorithm runs for $O(f)$ rounds (where $f \le t$ is the actual number of faulty processes) and honest processes send $O(n^2)$ messages.
\end{theorem}

The next theorem, which we do not formally prove in this paper, results from replacing the graded consensus employed in \Cref{theorem:early_stopping_ba_n3} with the authenticated graded consensus tolerating at most $t < n/2$ faulty processes from \Cref{theorem:gc_n2}---that is the only place where they assume $t < n/3$. 

\begin{theorem}[Authenticated Early-Stopping Byzantine Agreement]
    \label{theorem:early_stopping_ba_n2}
    There is an authenticated algorithm that solves Byzantine Agreement when there are at most $t < n/2$ faulty processes. The algorithm runs for $O(f)$ rounds (where $f \le t$ is the actual number of faulty processes) and honest processes send $O(n^2)$ messages.
\end{theorem}

We are now ready to prove the theorems we promised upfront in \Cref{section:intro}:

\begin{theorem}
    \label{theorem:unauth_ba_predicitions}
    There is an unauthenticated algorithm that solves Byzantine Agreement when there are at most $t < n/3$ faulty processes. Given a classification prediction with $B$ incorrect bits:
    \begin{itemize}
        \item If $B \le \beta n^{3/2}$ for some constant $\beta$, every honest process returns within $O(\min\{B/n, f\})$ rounds and, in total, honest processes send $O(n^2 \log(\min\{B/n, f\}))$ messages.
        \item Otherwise, every honest process returns within $O(f)$ rounds and, in total, honest processes send $O(n^2 \log(f))$ messages. 
    \end{itemize}
\end{theorem}
\begin{proof}
    In \Cref{algo:final_ba_copy}, we use \Cref{theorem:gc_n3} for each \textbf{graded-consensus}, \Cref{theorem:early_stopping_ba_n3} for each \textbf{ba-early-stopping}, and \Cref{theorem:unauth_ba_class} for each \textbf{ba-with-classification}.

    As graded consensus and early-stopping Byzantine Agreement tolerate $t < n/3$ faulty processes, the algorithm solves Byzantine Agreement when there are at most $t < n/3$ faulty processes.

    As $t < n/3$, by \Cref{Flemma:advice_preprocess_misclassify_bound} there will be $k_A \in O(B/n)$ misclassified processes from the result of classify($a_i$). By setting $\beta$ appropriately, if $B \le \beta n^{3/2}$, $(2k_A+1)(3k_A+1) \le n-t-k_A$. In this case, in a phase $\phi$ with $2^{\phi-1} \ge k_A$ and $ (2 \cdot 2^{\phi-1}+1)(3 \cdot 2^{\phi-1}+1) \le n-t-2^{\phi-1}$, the preconditions of \textbf{ba-with-classification} are satisfied. Here, the smallest $\phi$ is at most $O(\log_2(B/n))$.

    Next, the preconditions of \textbf{ba-early-stopping} are satisfied in phase $\phi$ if, during this phase, $T \in \Theta(\min\{2^{\phi-1}, t\})$ is set such that when there are $f \le t$ faulty processes, all honest processes return from \textbf{ba-early-stopping} in at most $T$ rounds. As with $f$ faulty processes all honest processes return in $O(f)$ rounds, the smallest $\phi$ is at most $O(\log_2(f))$.

    Combining the above observations with \Cref{lemma:final-ba-conditional-decide-early-ba} and \Cref{lemma:final-ba-conditional-decide-classification-ba}, all honest processes decide in $O(\log_2\min\{B/n, f\})$ phases when $B \le cn^{3/2}$, and $O(\log_2(f))$ phases otherwise. 
    To get the complexities, we plug in \Cref{lemma:final-ba-complexities} with $mc_{gc} = O(n^2)$, $mc_{early} = O(n^2)$, and $mc_{class} = O(n^2)$.
\end{proof}

\begin{theorem}
    \label{theorem:auth_ba_predicitions}
    There is an authenticated algorithm that solves Byzantine Agreement when there are at most $t < (1/2-\epsilon)n$ faulty processes. Given a classification prediction with $B$ incorrect bits, every honest process returns within $O(\min\{B/n, f\})$ rounds and, in total, honest processes send $O(n^3 \log (\min\{B/n, f\}))$ messages.
\end{theorem}
\begin{proof}
     In \Cref{algo:final_ba_copy}, we use \Cref{theorem:gc_n2} for each \textbf{graded-consensus}, \Cref{theorem:early_stopping_ba_n2} for each \textbf{ba-early-stopping}, and \Cref{theorem:auth_ba_class} for each \textbf{ba-with-classification}.


    As graded consensus and early-stopping Byzantine Agreement tolerates $t < n/2$ faulty processes, the algorithm solves Byzantine Agreement when there are at most $t < (1/2-\epsilon)n$ faulty processes. (Note that $(1/2-\epsilon)n < n/2$.)
    
    As $t < (1/2-\epsilon)n$, by \Cref{Flemma:advice_preprocess_misclassify_bound} there will be $k_A \in O(B/n)$ misclassified processes from the result of classify($a_i$). Suppose that $B \le \beta n^2$ for some constant $\beta$ such that $2k_A+1 \le n-t-k_A$.
    Then, in a phase $\phi$ with $2^{\phi-1} \ge k_A$ and $2 \cdot 2^{\phi-1}+1 \le n-t-2^{\phi-1}$, the preconditions of \textbf{ba-with-classification} are satisfied. Here, the smallest $\phi$ is at most $O(\log_2(B/n))$.

    Next, as stated in \Cref{theorem:unauth_ba_predicitions}, the earliest phase such that all honest processes output from the \textbf{ba-early-stopping} (and the output satisfies Agreement and Strong Unanimity), is phase $\phi$ where $\phi$ is at most $O(\log_2(f))$.

    From \Cref{lemma:final-ba-conditional-decide-early-ba}, \Cref{lemma:final-ba-conditional-decide-classification-ba}, and the fact that when $B \in \Omega(n^2)$ then $B/n \in \Omega(n)$, all honest processes decide in $O(\log_2\min\{B/n, f\})$ phases. By plugging in \Cref{lemma:final-ba-complexities} with $mc_{gc} = O(n^2)$, $mc_{early} = O(n^2)$, and $mc_{class} = O(n^3)$, all honest processes decide in $O(\min\{B/n, f\})$ rounds and they collectively send $O(n^3 \log(\min\{B/n, f\}))$ messages.
\end{proof}

%% file: lower-bounds.tex
\section{Lower Bounds}
\label{sec:lowerbounds}

In this section, we show two lower bounds: first, a lower bound on round complexity for algorithms that rely on classification predictions; second, a lower bound on message complexity that holds regardless of the type of predictions.

\subsection{Lower Bound on Round Complexity}\label{subsection:lower_bound_rounds}

The following result gives a lower bound on the round complexity as a function of the number of incorrect bits in the prediction.  Roughly, it states that we need at least $\Omega(\min(B/n+1, f))$ rounds, i.e., the dependence on $B$ for the algorithms presented in this paper is unavoidable.  This also implies that if the total number of incorrect bits in the \prediction{} is large enough, then
an algorithm with predictions cannot take fewer rounds
than an algorithm without predictions. 

More specifically, if the average number of incorrect bits per nonfaulty process, $B/(n-f)$, is at least $f$, then the \prediction{} might not identify any of the faulty processes and provides no benefit.
However, if the number of incorrect bits in the \prediction{} is smaller, then the number of non-faulty
processes that are not identified in the \prediction{} is at most $B/(n-f)$ and the number of other faulty processes is at least 
$x = f - B/(n-f)$. In this case, an algorithm with predictions cannot take fewer rounds than an algorithm
without predictions having $x$ fewer processes and $x$ fewer faults.

\begin{theorem}
\label{theorem:rndLowerBound}
    For every deterministic Byzantine agreement algorithm with classification predictions
    and for every $f \leq t < n-1$, there is an execution with $f$ faults in which at least
    $\min\{f+2,t+1, \lfloor B/(n-f) \rfloor+2,\lfloor B/(n-t) \rfloor+1\}$ rounds are performed.
\end{theorem}

\begin{proof}
Let $A$ be a consensus algorithm with predictions for $n$ processes $p_1,\ldots,p_{n}$ which tolerates $t< n-1$ faulty processes and at most $B$ bits of wrong predictions. Suppose it terminates within $T(f)$ rounds, when the actual number of faulty processes is $f \leq t$.

First suppose that $B \geq f(n-f)$.
Consider the consensus algorithm $S$ without predictions in which processes perform algorithm
$A$, assuming that all processes are predicted to be honest.
In an execution in which $f \leq t$ of these processes are faulty, 
the total number of faulty processes being simulated is $f$
and the total number of non-faulty processes being simulated is $n-f$.
In this execution, the total number of faulty bits of predictions given to non-faulty processes
is $(n-f)\cdot f \leq B$.
Since $A$ terminates within $T(f)$ rounds under these conditions,
so does algorithm $S$.
Since any deterministic consensus algorithm without predictions that tolerates at most $t$ faulty processes 
has an execution with $f \leq t$ faulty processes that takes
at least $\min\{f+2,t+1\}$ rounds \cite{DRS90}, it follows that $T(f) \geq \min\{f+2,t+1\}$.

Now suppose that $B < f(n-f)$. Then $B/(n-f) < f$.
Let $x = f - \lfloor B/(n-f) \rfloor \geq 1$, let $n' = n -x$, and let $t' = t -x$.
Consider the following consensus algorithm $S'$ without predictions for $n'$ processes, $q_1,\ldots,q_n$, at most $t' = t - (n-n')$ of which may fail by crashing:
processes $q_i$ simulates process $p_i$ for $1 \leq i \leq n'$ with the predictions that
$p_1,\ldots,p_{n'}$ are honest and the last $x$ processes, $p_{n'+1},\ldots,p_{n}$ are faulty.
Note that process $q_i$ will receive no messages that appear to be sent by processes $p_{n'+1},\ldots..,p_{n}$, so it will treat these $x = n-n'$ processes as if they crashed at the beginning of every execution.

If process $q_i$ crashes, 
then the corresponding process
$p_i$ being simulated also crashes.
In any execution of $S'$ in which $f' = f -x \leq t-x = t'$ of these processes are faulty,
the total number of faulty processes being simulated is $f = x + f'$
and the total number of non-faulty processes being simulated is $n' - f' = n - f$.
So, the total number of faulty bits of predictions given to non-faulty processes is
$(n'-f')f' = 
(n-f)(f-x) = (n-f)\lfloor B/(n-f) \rfloor \leq B$.
Since $A$ terminates within $T(f)$ rounds under these conditions,
so does algorithm $S'$.
Any deterministic consensus algorithm without predictions that tolerates at most $t'$ faulty processes 
has an execution with $f' \leq t'$ faulty processes that takes
at least $\min\{f'+2,t'+1\}$ rounds. In particular, this is true for $S'$.
If $f' = t'$, then $f =t$ and $T(t) \geq t'+1 = t - x + 1 = \lfloor B/(n-t) \rfloor + 1$.
If $f' < t'$, then $T(f) \geq f'+2 = f - x + 2 = \lfloor B/(n-f) \rfloor + 2$.
\end{proof}

\subsection{Lower Bound on Message Complexity}


In this section, we show that every deterministic, synchronous Byzantine Agreement protocol with predictions requires $\Omega(n^2)$ in the worst-case, even in executions where the predictions are perfectly accurate.  This lower bound holds even if messages are authenticated.

The basic idea of the proof is quite similar to the classical lower bound by Dolev and Reischuk~\cite{dolev1985bounds}.  
As in the Dolev-Reischuk paper, the proof here considers Byzantine broadcast instead of Byzantine agreement. The problem of Byzantine broadcast assumes a specified sender (i.e., the broadcaster), and requires that: (i) all honest processes deliver the same message (agreement), and (ii) if the sender is honest, all honest processes deliver the sender's message (validity).  

We note (by well-known reduction) that if Byzantine broadcast requires $\Omega(n^2)$ messages, then Byzantine agreement does too: if you already have an efficient Byzantine agreement protocol, then Byzantine broadcast can be solved by having the sender first send its message to everyone (with $O(n)$ messages), and then all processes execute a Byzantine agreement protocol on the message to deliver.

The key difficulty lies in modifying the proof to hold even when there are predictions---and to construct an execution with perfectly accurate predictions that has high message complexity.  The proof follows the typical paradigm of dividing the processes into two sets and constructing indistinguishable executions.  We also need to specify/construct the predictions that the processes will receive in such a way as to ensure the desired outcome.


\begin{theorem}
    \label{theorem:message_lower_bound}
    For every deterministic algorithm that solves Byzantine broadcast with predictions in a synchronous system where up to $t$ processes may be faulty, there is some execution in which the predictions is $100\%$ correct and at least $\Omega(n + t^2)$ messages are sent by honest processes.  This implies that the same bounds hold for Byzantine agreement.
\end{theorem}

\begin{proof}

We first show the straightforward claim that every protocol requires $\Omega(n)$ messages.  Consider the executions $E_i$ in which the sender sends the value $i$, each process is honest, and each process gets correct predictions, for $i \in \{0,1\}$.
Suppose that, in both these executions, fewer than $\lceil n/4\rceil$ messages are sent. Then there exists a process $q$ that does not send or receive a message in either execution.
Now consider the execution $E'$, in which the sender sends the value 1, all processes are honest, all processes except $q$ get the same predictions as in $E_1$,
and process $q$ get the same predictions as it gets in $E_0$.
All processes send and receive the same messages in $E'$ as they do in $E_1$. Since all processes except $q$ get the same predictions as in $E_1$, they all
decide 1. However, $q$ sends and receives the same messages in $E'$ as it does in $E_0$, so it decides 0.
Therefore $E'$ violates agreement. This means that at least $\lceil n/4 \rceil$ messages must have been sent in either $E_0$ or $E_1$.

We now focus on the dependence on $t$.  Let $S_0$ be the set of processes $p$, other than the sender, for which there is some prediction string $a$ such that $p$ eventually decides $0$ if it receives predictions $a$, but no messages.
Symmetrically, let $S_1$ be the set of processes, other than the sender, that decide $1$ for some prediction string when they receive no messages. Clearly, at least one of these sets must contain at least $\lceil (n-1)/2\rceil$ processes. Without loss of generality, assume $\vert S_0 \vert \geq \lceil (n-1)/2\rceil$.

Let $B\subseteq S_0$ be a set of $\lfloor t/2 \rfloor$ processes. Consider the following execution,
\textsc{$E_{good}$}. In this execution, all processes receive correct predictions. Processes in $B$ are faulty and all other processes are honest. The sender is honest and receives 1 as input.
Faulty processes behave as follows: they ignore the first $\lfloor t/2 \rfloor$ messages sent to them and do not send any messages to each other. Other than that, they follow the algorithm (in that they send messages as per the algorithm), pretending to have received predictions that would make them decide 0 when they do not receive any messages.
By validity, all honest processes must decide 1 in this execution.

Consider the number of messages sent by honest processes to processes in $B$ in $E_{good}$. If each process in $B$ received at least $\lceil t/2 \rceil$ messages from honest processes in $E_{good}$, then $E_{good}$ is an execution in which $\Omega(t^2)$ messages are sent. Otherwise, let $p\in B$ be a process that received fewer than $\lceil t/2 \rceil$ messages from honest processes in $E_{good}$. Let $A_p$ be the set of honest processes that sent messages to $p$.

\medskip

Now consider the following execution,
\textsc{$E_{bad}$.} Processes in $A_p$ and $B - \{p\}$ are faulty. All other processes are honest. Note that $\vert A_p \vert + \vert  B - \{ p \} \vert < t$. 
Process $p$ receives predictions that would make it decide 0 when it does not receive any messages.
All other processes receive the same predictions they did in $E_{good}$. 

Processes in $A_p$ behave the same as in $E_{good}$ except that they don't send any messages to $p$. Processes in $B\setminus \{ p \}$ behave the same as they did in $E_{good}$.

Note that, in this execution, process $p$ receives the same predictions as it pretended to receive in $E_{good}$.
In $E_{bad}$,  it receives no messages, whereas
in $E_{good}$, it behaved as if it had received no messages.
Thus, the behavior of process $p$ in $E_{bad}$ is identical to its behavior in $E_{good}$.
Hence, in $E_{bad}$, process $p$, which is honest, decides 0.
For all other honest processes, execution $E_{bad}$ is indistinguishable from $E_{good}$. Therefore, they all decide 1, as they did in $E_{good}$.
Therefore $E_{bad}$ violates agreement. This means that at least $\lceil t/2 \rceil$ messages must have been sent to each process in $B$ in execution $E_{good}$ (which had all correct predictions).
\end{proof}

\paragraph{Observations.} The above lower bound does not make any assumptions on the type of predictions processes receive or whether or not the processes can use cryptographic signatures. In this sense, it is quite a strong bound. However, it does assume that different processes may receive different predictions in the same execution. It therefore relies on \emph{non-uniform} predictions.  If processes were guaranteed to receive the same predictions, then it remains possible that one could beat this lower bound.

%% file: conclusion.tex
\section{Conclusion}\label{section:conclusion}


In this paper, we study Byzantine agreement with predictions. Surprisingly, predictions do not help in reducing the message complexity. Meanwhile, using classification predictions, we showed that Byzantine agreement can be solved in $O(\min\{B/n+1, f\})$ rounds, where $B$ is the number of incorrect bits in the prediction. This is asymptotically tight, as we also show a lower bound of $\Omega(\min\{B/n, f\})$ rounds.

These results also open up several exciting research directions. First, while our algorithm has asymptotically optimal round complexity, we did not try to optimize its communication complexity (i.e. bits sent by honest processes). Notably, our voting step already induces $O(n^3)$ communication complexity. It would be interesting if the communication complexity can be reduced, while still maintaining the round complexity. Second, our unauthenticated algorithm can only productively use prediction with $B = O(n^{3/2})$ incorrect bits. It would be interesting if this could be pushed to $B = O(n^2)$ as in the authenticated case. Finally, it would be interesting to study other types of predictions, or how predictions can be used in different models of distributed algorithms (e.g., asynchronous networks).